\documentclass[a4paper]{article}

\usepackage[utf8]{inputenc}

\usepackage[dvipdfmx]{graphicx}
\usepackage{bmpsize}
\usepackage{multicol}
\usepackage{amsthm}
\usepackage{bm}
\usepackage{amssymb}
\usepackage{amsmath}
\usepackage{color}
\usepackage{amsthm}
\usepackage{subcaption} 
\usepackage{graphicx}
\usepackage{dcolumn}
\usepackage{bm}

\usepackage[utf8]{inputenc}
\usepackage[T1]{fontenc}
\usepackage{etoolbox}
\usepackage{ascmac,wrapfig,makeidx}
\usepackage{physics}
\usepackage[stable]{footmisc}
\usepackage{mathrsfs}
\usepackage{color}
\usepackage{comment}
\usepackage{amsmath}
\usepackage{physics}

\theoremstyle{plain}
\newtheorem{theorem}{Theorem}
\newtheorem*{theorem*}{Theorem}
\newtheorem{lemma}{Lemma}[section]

\newtheorem{proposition}{Proposition}[section]

\theoremstyle{definition}

\newcommand{\BA}{\begin{eqnarray}}
\newcommand{\EA}{\end{eqnarray}}

\definecolor{dgreen}{rgb}{0.0, 0.5, 0.0}

\begin{document}

\fontsize{14pt}{16.5pt}\selectfont

\begin{center}
\bf{Rigged Liouville space formulation for\\
quasi-Hermitian Liouville operators
}
\end{center}
\fontsize{12pt}{11pt}\selectfont
\begin{center}
Shousuke Ohmori$^{1,*)}$ and Junichi Takahashi$^{2)}$ 
\end{center}

\bigskip

\noindent
\it{
1)Department of Economics, Hosei University, Machida-shi, Tokyo 194-0298, Japan.
}\\
\it{
2)~Department of Economics, Asia University, 
Musashino-shi, Tokyo 180-0022, Japan.}

\bigskip

\noindent
*corresponding author: 42261timemachine@ruri.waseda.jp\\
~~\\
\rm
\fontsize{11pt}{14pt}\selectfont\noindent

\baselineskip 22pt

{\bf Abstract}\\
%
We discuss a super bra-ket formalism for quasi-Hermitian Liouvillian operators within the framework of rigged Hilbert spaces (RHS). 
An RHS in terms of the Liouville space, referred to as a rigged Liouville space (RLS), is reconstructed by exploiting the mathematical fact that the space of Hilbert-Schmidt operators is unitarily equivalent to the tensor product of Hilbert spaces.
The obtained RLS endows a rigorous foundation of the construction for the super bra-ket and for the spectral decompositions of both Hermitian and quasi-Hermitian Liouville operators, which are characterized by the generalized eigenvectors in the dual spaces.
Furthermore, within this framework, the non-Hermitian Liouvillian operator and its adjoint can be constructed symmetrically, with their symmetric structure preserved.
As an application of our RLS methodology, we examine the Liouville operators corresponding to Hermitian and non-Hermitian harmonic oscillators and elucidate the differences between their spectral decomposition forms.


\section{Introduction}
\label{sec:1}

Rigged Hilbert space (RHS) formalism has been developed to address the issue of constructing the mathematical foundation of quantum mechanics described based on Dirac's bra-ket notation\cite{Robert1966a,Robert1966b,Antoine1969a,Antoine1969b,Melsheimer1974a,Melsheimer1974b,Bohm1978,Bohm1981,Gadella1983,Gadella1984,Bohm1989,Prigogine1996,Bohm1998,Antoiou1998a,Antoiou1998b,Antoiou2003,Gadella2003,Chruscinski2003,Chruscinski2004,Madrid2004,Madrid2005,Antoine2009,Liu2013,Celeghini,Antoine2021,Fernandez2022,Ohmori2022,Ohmori2024}.
The key of this formalism is the triplet of the topological vector spaces constructed by combining a Hilbert space $\mathcal{H}=(\mathcal{H}, \langle \cdot, \cdot \rangle_\mathcal{H})$ with its dense linear subspace $\Phi$ and its dual space $\Phi^\prime$, 
\begin{eqnarray}
    \Phi\subset \mathcal{H} \subset\Phi^{\prime},
    \label{eqn:RHS0}
\end{eqnarray}
where the subspace $\Phi$ has a nuclear topology, namely, $\Phi=(\Phi, \tau_\Phi)$ is a nuclear space, such that the inner product $\langle \cdot, \cdot \rangle_\Phi$ on $\Phi$, $\langle \phi, \psi \rangle_\Phi \equiv \langle \phi, \psi \rangle_\mathcal{H}$ for $\phi, \psi \in \Phi$, is separately continuous on $(\Phi, \tau_\Phi)$.
$\Phi^\prime$ is a family of continuous linear functionals on $(\Phi,\tau_\Phi)$.
The RHS (\ref{eqn:RHS0}) enables the mathematical treatment of not only the bra and ket vectors but also spectral decomposition characterized by distributions, such as the Dirac's $\delta$ function, appearing in Dirac's bra-ket formalism, which are insufficient to handle within the traditional Hilbert-space framework of von Neumann.
Actually, for a given Hermitian operator (observable), the nuclear spectral theorem\cite{Maurin1968} based on the RHS provides eigenequations for the bra and ket vectors, individually, and the spectral expansion for the Hermitian operator that is identified with the spectral decomposition by the discrete and continuous spectra described using
the Dirac's $\delta$-function found in the literature.
Nowadays, the methodology of RHS has been widely applied to several physical systems beyond the scope of traditional Hermitian quantum mechanics, such as non-equilibrium open systems\cite{Chruscinski2003,Chruscinski2004}, non-Hermitian quantum systems\cite{Ohmori2022,Ohmori2024}, and nonlinear dynamical systems\cite{Antoniou1993,Suchanecki1996,Chiba2015}.

Considering the symmetric structure of the bra and ket vectors, the approach using the following type of RHS, proposed by R. Madrid\cite{Madrid2005}, has also been developed.
\begin{eqnarray}
    \Phi\subset \mathcal{H} \subset\Phi^{\prime}, \Phi^\times.
    \label{eqn:RHS}
\end{eqnarray}
In this formulation, the space $\Phi^\times$ that is a family of continuous {\it anti-linear} functionals on $(\Phi,\tau_\Phi)$ is added to the RHS (\ref{eqn:RHS0})
(the function $f\in \Phi^{\times}$ is called anti-linear when satisfying $f(a\varphi+b\phi)=a^*f(\varphi)+b^*f(\phi)$, where $a$ and $b$ are complex numbers with complex conjugates $a^*$ and $b^*$ and $\varphi, \phi \in \Phi$).
Then, the bra and ket vectors are characterized by elements of $\Phi^\prime$ and $\Phi^\times$;
for $\varphi \in \Phi$, 
a map $\ket{\varphi}_\mathcal{H} : \Phi \rightarrow \mathbb{C}^1$, $\ket{\varphi}_\mathcal{H}(\phi) \equiv  \langle \phi,\, \varphi\rangle_\mathcal{H}$ for $\phi \in \Phi$ is called a ket vector of $\varphi$.
The bra vector of $\varphi$ is given as the complex conjugate of $\ket{\varphi}_\mathcal{H}$, 
namely, the map $\bra{\varphi}_\mathcal{H} : \Phi \rightarrow \mathbb{C}^1 $ where $\bra{\varphi}_{\mathcal{H}}(\phi)= \ket{\varphi}_{\mathcal{H}}^*(\phi)=(\ket{\varphi}_{\mathcal{H}}(\phi))^*=\langle \varphi,\, \phi\rangle_\mathcal{H}$.
%
%
Note that $\bra{\varphi}_{\mathcal{H}}$ and $\ket{\varphi}_{\mathcal{H}}$ belong to $\Phi^{\prime}$ and $\Phi^{\times}$, respectively.
The eigenequations and the spectral expansions are then formulated in the dual and anti-dual spaces $\Phi^{\prime}$ and $\Phi^{\times}$ by generalized eigenvectors corresponding to given Hermitian operator.
In this respect, the handling of bra and ket vectors as elements of these dual spaces constitutes one of the crucial features of the RHS formalism.
Hereafter, we refer to the combination of the dual space $\Phi^{\prime}$ and the anti-dual space $\Phi^{\times}$ as the dual spaces.

Based on this viewpoint of the RHS formalism, we have studied the mathematical description of the bra-ket formalism found in non-Hermitian quantum systems whose non-Hermitian operator ${A}$ is characterized by the  symmetric relation,
\begin{equation}
{A}^\dagger =\eta {A} \eta^{-1},
    \label{eqn:O1}
\end{equation}
where $\eta$ is the intertwining operator\cite{Ohmori2022,Ohmori2024}.
This type of operator is referred to as an $\eta$-quasi Hermitian\cite{Dieudonne1961,Mos2010,Antoine2013}
and found in several non-Hermitian quantum systems, such as 
$\mathcal{PT}$-symmetrical systems\cite{Bender1998,Bender2007}.
Focusing on the case where the intertwining operator $\eta$ is a positive operator, we presented a general framework for the bra-ket description using the RHS given as
\begin{equation}
    \Phi \subset \widetilde{\mathcal{H}_\eta} \subset \Phi^\prime,\Phi^\times,
    \label{eqn:RHS_eta}
\end{equation}
where $\widetilde{\mathcal{H}_\eta}=(\widetilde{\mathcal{H}_\eta},\widetilde{\langle \cdot, \cdot \rangle_\eta})$ is the completion of the pre-Hilbert space $\mathcal{H}_\eta=(\mathcal{H},\langle \cdot, \cdot \rangle_\eta)$ equipping the metric induced from the inner product $\langle \phi, \psi \rangle_\eta
    = \langle \phi, \eta \psi \rangle_\mathcal{H}$ for $\phi,\psi \in \mathcal{H}$ with respect to $\eta$\cite{Ohmori2022}.
This RHS (\ref{eqn:RHS_eta}) suitable for the $\eta$-quasi Hermitian is called an $\eta$-RHS.

In the $\eta$-RHS, an $\eta$-quasi Hermitian operator becomes Hermitian, and the bra-ket formalism can be formulated in the dual spaces. 
Furthermore, in the analysis of non-Hermitian composite systems, this framework ensures that the composite operator and its adjoint are well defined in a symmetric manner, while preserving their symmetric structure (\ref{eqn:O1})\cite{Ohmori2024}.
To be more specific, the following issue arises. 
For non-Hermitian (not self-adjoint) operators $A_i$ ($i = 1, 2$) acting on Hilbert spaces, the relation
$
(A_1^\dagger \otimes I_2 + I_1 \otimes A_2^\dagger) \neq
(A_1 \otimes I_2 + I_1 \otimes A_2)^\dagger
$
holds in general; only the inclusion
$(A_1^\dagger \otimes I_2 + I_1 \otimes A_2^\dagger) \subset
(A_1 \otimes I_2 + I_1 \otimes A_2)^\dagger
$
is satisfied.\cite{Simon1980}
This implies that for the composite operator $A=\overline{A_1\otimes I+I\otimes A_2}$, one generally has
$A^\dagger \neq \overline{A_1^\dagger \otimes I + I \otimes A_2^\dagger}$,
leading to an inconsistency in defining the adjoint operator of the composite system. 
However, this inconsistency can be resolved by extending these operators to the dual spaces within the RHS framework.
Therefore, the RHS formulation offers a particular advantage in addressing the problem of defining adjoint operators for non-Hermitian composite systems.
A similar issue arises in the treatment of non-Hermitian Liouvillian operators, as will be discussed in Sec.~\ref{sec:4.3}.

The present study aims to apply our previous framework to a non-Hermitian Liouville operator acting on Liouville space.
There have been several previous studies applying RHS formulation to Liouville space so far\cite{Antoiou1998a,Antoiou1998b,Antoiou2003,Liu2013}.
In these studies, the concept of an RHS suitable for Liouville space, called rigged Liouville space (RLS), was introduced and the general properties of the Liouville operator, including its spectral decomposition were discussed.
%
%
(The precise definition of RLS is given in the present study).
The previous RLS treatments mentioned above have been primarily limited to the
case where the Liouville operator is Hermitian (self-adjoint). 
However, in recent studies of non-equilibrium open systems and non-Hermitian quantum systems,
the role of the non-Hermitian Liouville operator has become increasingly
important\cite{Gyamfi2020, Polonyi2021, Sukharnikov2023}. 
For instance, in the Lindblad master equation that is well-known as the representative equation governing open quantum systems, the generator (Lindbladian) acts as a non-Hermitian superoperator on the Liouville space, and its spectral properties play a central role in understanding the relaxation dynamics of the system\cite{Wang2026}.
As a first step toward formulating the RHS for the non-Hermitian Liouvillian
operator, the present study is devoted to investigating the super bra-ket
formulation governed by a quasi-Hermitian Liouville operator.

In the next section, we offer a general construction of the RLS by utilizing the mathematical fact that Liouville space can be unitarily transformed into the tensor product of Hilbert spaces.  
A major advantage of this construction is that it clarifies how the super bra-ket vectors and spectral expansions are described on the dual spaces within the RLS framework.
Furthermore, it can be readily extended to the quasi-Hermitian case.
Using the obtained RLS, we fully present the mathematical formulation of the spectral expansion of the super bra and ket vectors for the Hermitian Liouville operator in Section~\ref{sec:3}. 
As a practical example, we apply our formulation to the harmonic oscillator at the end of the section.
Section~\ref{sec:4} focuses on the quasi-Hermitian Liouville operator and investigates the mathematical structure of its spectral expansions by the corresponding super bra and ket vectors.
We then show that, in executing the super bra-ket formalism under the dual spaces, the adjoint of a given non-Hermitian Liouvillian operator can be consistently defined, and the corresponding spectral expansions of the super bra and ket vectors can be formulated.
The obtained results are applied to the non-Hermitian harmonic oscillator to clarify the essential  differences from the Hermitian case.
Finally, the conclusions are given in Section~\ref{sec:6}.
%


\section{Construction of a rigged Liouville space  induced from the tensor product of rigged Hilbert spaces}
\label{sec:2}

We now present construction of a RLS from the perspective of the unitary transformation between the Liouville space and a tensor product of Hilbert spaces.
We recall that the Liouville space is given as the Hilbert space, $\mathcal{L}(\mathcal{H})=(\mathcal{L}(\mathcal{H}),\langle \cdot,\, \cdot\rangle_{\mathcal{L}})$, composed of the set of Hilbert–Schmidt operators on the Hilbert space $\mathcal{H}$ with the inner product $\langle A,\, B\rangle_{\mathcal{L}}=\Tr{A^{\dagger}B}=\sum_{i}\langle \varphi_i,\, A^\dagger B\varphi_i\rangle _\mathcal{H}$ where $\{\varphi_i\}$ is an orthonormal base of $\mathcal{H}$\cite{Prugovecki1981}.
Set the tensor product of a given RHS, $\Phi\subset \mathcal{H} \subset\Phi^\prime,\Phi^\times$, by
\begin{eqnarray}
    \Phi \hat{\otimes} {\Phi} \subset \mathcal{H} \bar {\otimes} \mathcal{H} \subset 
    (\Phi \hat{\otimes} {\Phi})^{\prime},~(\Phi \hat{\otimes} {\Phi})^{\times},
    \label{eqn:tensor_pro_of_RHS}
\end{eqnarray}
where $\mathcal{H} \bar {\otimes} \mathcal{H}=(\mathcal{H} \bar{\otimes} \mathcal{H},\langle \cdot,\, \cdot\rangle_{\mathcal{H} \bar{\otimes} \mathcal{H}})$ is the completion of the algebraic tensor product $\mathcal{H} \otimes \mathcal{H}=\Big{\{}\displaystyle\sum_{j=1}^m\varphi_{1j}\otimes\varphi_{2j}\mid \varphi_{1j},\varphi_{2j}\in \mathcal{H},j=1\sim m,m\in \mathbb{N} \Big{\}}$ for $\mathcal{H}=(\mathcal{H}, \langle \cdot, \cdot \rangle_\mathcal{H})$ equipping the topology induced by an inner product $\langle \cdot,\, \cdot\rangle_{\mathcal{H} \otimes \mathcal{H}}$, satisfying $\langle \varphi_1 \otimes \varphi_2,\, \phi_1\otimes\phi_2\rangle_{\mathcal{H} \otimes \mathcal{H}}=\langle \varphi_1 ,  \phi_1 \rangle_{\mathcal{H}}\langle \varphi_2 ,  \phi_2 \rangle_{\mathcal{H}}$.
The space $\Phi \hat{\otimes} {\Phi}=(\Phi_1 \hat{\otimes} {\Phi}_2,\hat{\tau}_p)$ is the completion of the tensor product $(\Phi \otimes {\Phi},\tau_p)$ where $\Phi \otimes {\Phi}=\Big{\{}\displaystyle\sum_{j=1}^m\varphi_{1j}\otimes\varphi_{2j}\mid \varphi_{1j},\varphi_{2j}\in {\Phi},j=1\sim m,m\in \mathbb{N} \Big{\}}$ and $\tau_p$ is a locally convex topology with the local base $\mathcal{B}_p=\{\Gamma(V\otimes V) \mid V\in \mathcal{B}\}$ where $\mathcal{B}$ is a local base of the nuclear space $(\Phi,\tau_{\Phi})$ and $\Gamma(X)$ stands for the convex circled hull of a set $X$\cite{Schaefer1966}.
The spaces $(\Phi \hat{\otimes} {\Phi})^\prime$ and $(\Phi \hat{\otimes} {\Phi})^\times$ are dual and anti-dual spaces of $\Phi \hat{\otimes} {\Phi}$, respectively.
Note that $\Phi \hat{\otimes} {\Phi}$ becomes a nuclear space and the triplet (\ref{eqn:tensor_pro_of_RHS}) is an RHS\cite{Maurin1968}.
The bra and ket vectors corresponding to $\varphi\in \Phi \hat{\otimes} \Phi$ are defined by
\begin{eqnarray}
    \bra{\varphi}_{\mathcal{H} \bar{\otimes} \mathcal{H}} : \Phi \hat{\otimes} {\Phi} \to \mathbb{C},\phi \mapsto \langle \varphi,\, \phi\rangle_{\mathcal{H} \bar {\otimes} \mathcal{H}}, 
    \label{eqn:bra_of_tensor_pro_of_RHS}\\
    \ket{\varphi}_{\mathcal{H} \bar{\otimes} \mathcal{H}} : \Phi \hat{\otimes} {\Phi} \to \mathbb{C},\phi\mapsto \langle \phi,\, \varphi\rangle_{\mathcal{H} \bar{\otimes} \mathcal{H}}.    
    \label{eqn:ket_of_tensor_pro_of_RHS}
\end{eqnarray}
Here, the relations $\ket{\varphi}_{\mathcal{H} \bar{\otimes} \mathcal{H}}=\bra{\varphi}_{\mathcal{H} \bar{\otimes} \mathcal{H}}^*$, $\bra{\varphi}_{\mathcal{H} \bar{\otimes} \mathcal{H}} \in (\Phi \hat{\otimes} {\Phi})^{\prime}$, and $\ket{\varphi}_{\mathcal{H} \bar{\otimes} \mathcal{H}} \in (\Phi \hat{\otimes} {\Phi})^{\times}$, are satisfied.
Note that for $\varphi\otimes \psi\in \Phi\otimes\Phi\subset \Phi\hat{\otimes}\Phi$, we have the following relations\cite{Ohmori2023}:
\begin{eqnarray}
    \bra{\varphi \otimes \psi}_{\mathcal{H} \bar{\otimes} \mathcal{H}}=\bra{\varphi}_{\mathcal{H}} \otimes \bra{\psi}_{\mathcal{H}},
    ~~\ket{\varphi \otimes \psi}_{\mathcal{H} \bar{\otimes} \mathcal{H}}=\ket{\varphi}_{\mathcal{H}} \otimes \ket{\psi}_{\mathcal{H}}.
\label{eqn:braket_tensor_relation}
\end{eqnarray}
To associate the RHS (\ref{eqn:tensor_pro_of_RHS}) with the Liouville space, $\mathcal{L}(\mathcal{H})=(\mathcal{L}(\mathcal{H}),\langle \cdot,\, \cdot\rangle_{\mathcal{L}})$, we consider  a conjugation $C:\mathcal{H}\to \mathcal{H}$, namely, it satisfies (i) $C$ is anti-linear, (ii) $C^2=I$, and (iii) $\|C\varphi\|_{\mathcal{H}}=\|\varphi\|_{\mathcal{H}}~(\varphi\in \mathcal{H})$.
Note that $C$ is bijective.
Then, there exists the unitary operator,
%
%
$I_C :  (\mathcal{H} \bar{\otimes} \mathcal{H},\langle \cdot,\, \cdot\rangle_{\mathcal{H} \bar{\otimes} \mathcal{H}}) \to (\mathcal{L}(\mathcal{H}),\langle \cdot,\, \cdot\rangle_{\mathcal{L}}),$
which satisfies 
\begin{eqnarray}
I_C(\varphi\otimes C\psi) = P_{\psi,\varphi}
    \label{eqn:LandT_unitary_relation}
\end{eqnarray}
for $\varphi,\psi\in\mathcal{H}$, where
$P_{\psi,\varphi} : \mathcal{H}\to \mathcal{H},\phi\mapsto \langle \psi,\, \phi\rangle _{\mathcal{H}}\varphi$.
%
Note that $P_{\psi,\varphi}$ is often denoted by $\ket{\varphi}\bra{\psi}$.
However, we do not use the notation in this paper to avoid confusion with the bra-ket notation.

In establishing the RLS based on the above unitary operator $I_C$,
we here endow the general construction of RHS in the situation where for given RHS, $\Phi_1 \subset \mathcal{H}_1 \subset \Phi_1^{\prime},\Phi_1^{\times}$, the Hilbert space $\mathcal{H}_1=(\mathcal{H}_1,\langle \cdot, \cdot \rangle_1)$ is associated with another Hilbert space $\mathcal{H}_2=(\mathcal{H}_2,\langle \cdot, \cdot \rangle_2)$ by an unitary operator $U$.
\begin{proposition}
    \label{proposition2.1}
    Let $\Phi _1 \subset \mathcal{H}_1 \subset \Phi^\prime_1,\Phi^\times_1$ be an RHS  and let $\mathcal{H}_2$ be a Hilbert space such that there is a unitary transformation $U$ from $\mathcal{H}_1$ onto $\mathcal{H}_2$.
Then, there is a linear dense subspace $\Phi_2$ of $\mathcal{H}_2$, equipping the nuclear topology $\tau_2$, and the triplet $\Phi _2 \subset \mathcal{H}_2 \subset \Phi^\prime_2,\Phi^\times_2$ is an RHS.
Furthermore, the restriction $T:=U_{\Phi_1}$ of the unitary transformation $U$ to the nuclear space $\Phi_1$ is isomorphic (i.e., linear homeomorphic) onto the nuclear space $(\Phi_2,\tau_2)$.
\end{proposition}

\begin{proof}
    By hypothesis, for the nuclear space $\Phi_1=(\Phi_1,\tau_1)$, the including map $i:(\Phi_1,\tau_1) \rightarrow (\mathcal{H}_1,\langle \cdot, \cdot \rangle_1), \varphi \mapsto \varphi$ is continuous and $\Phi_1$ is dense in $(\mathcal{H}_1,\langle \cdot, \cdot \rangle_ {1})$.
    We first assume the existence of a topological vector space $(\Phi_2,\tau_2)$ of a subspace of $\mathcal{H}_2$ such that the restriction $T:=U_{\Phi_1} : (\Phi_1,\tau_1) \rightarrow (\Phi_2,\tau_2)$ is isomorphic.
    Then, the topological vector space $(\Phi_2,\tau_2)$ becomes nuclear.
    To show it, let $\mathcal{K}$ be a Hilbert space.
    Since $(\Phi_1,\tau_1)$ is nuclear, there is a nuclear mapping $u:(\Phi_1,\tau_1) \rightarrow \mathcal{K}$\cite{Schaefer1966}.
    Then, it is confirmed that the composition $u\circ T^{-1}:(\Phi_2,\tau_2) \rightarrow \mathcal{K}$ of $u$ and $T^{-1}$ becomes a nuclear mapping, which implies that $(\Phi_2,\tau_2)$ is a nuclear space.
    Such nuclear space $(\Phi_2,\tau_2)$ is dense in $(\mathcal{H}_2,\langle \cdot, \cdot \rangle_ 2)$.
    This fact can be easily proved because $\Phi_1$ is dense in $\mathcal{H}_1$ and $U$ is the unitary transformation satisfying $\Phi_2=T(\Phi_1)=U(\Phi_1)$.  
    Note that the including mapping $i' : (\Phi_2,\tau_2) \rightarrow (\mathcal{H}_2,\langle \cdot, \cdot \rangle_ 2)$ is continuous, since $i'=U\circ i \circ T^{-1}$ holds.
    Thus, the triplet $\Phi_2\subset \mathcal{H}_2\subset \Phi_2^{\prime},\Phi_2^{\times}$ constructs an RHS.
    Finally, we show the existence of $(\Phi_2,\tau_2)$.
    Let $\Phi_2=U(\Phi_1)$; it is a linear subspace of $\mathcal{H}_2$.
    Now, define a map $T:(\Phi_1,\tau_1)\rightarrow \Phi_2$ as $T(\varphi)=U\varphi$, which is linear bijective.
    Let $\tau_2$ be the induced topology from $\tau_1$ by $T$ where $\tau_2$ has the local base $\beta_2 = \{W\subset \Phi_2 \mid T^{-1}(W)\in \beta _1 \}$ ($\beta _1$ is a local base of $\tau _1$).
    Then, it is easy to confirm that $T:(\Phi_1,\tau_1) \rightarrow (\Phi_2,\tau_2)$ is isomorphic. 
\end{proof}
\noindent
We assume that $C$ is continuous on $(\Phi,\tau_\Phi)$ and $C(\Phi)=\Phi$.
The latter assumption gives for $\varphi,\psi\in\Phi$, $\varphi\otimes C\psi\in \Phi\hat{\otimes}\Phi$.
From Proposition \ref{proposition2.1}, we can obtain the nuclear space, $(\Phi_\mathcal{L},\tau_{\Phi_\mathcal{L}})$, where  $\Phi_\mathcal{L}=I_C(\Phi\hat{\otimes}\Phi)$ and the nuclear topology $\tau_{\Phi_\mathcal{L}}$ possesses the local base $\mathcal{B}=\{W\subset \Phi_\mathcal{L}\mid I_C^{-1}(W)\in\mathcal{B}_{\Phi\hat{\otimes}\Phi}\}$~($\mathcal{B}_{\Phi\hat{\otimes}\Phi}$ is the local base of $\Phi \hat{\otimes}\Phi$).
Then, the triplet 
\begin{eqnarray}
    {\Phi}_\mathcal{L} \subset \mathcal{L}(\mathcal{H}) \subset 
    {\Phi}_\mathcal{L}^{\prime},{\Phi}_\mathcal{L}^{\times}
    \label{eqn:RLS}
\end{eqnarray}
consists of an RHS of the form (\ref{eqn:RHS}).
Hereafter, we call this RHS (\ref{eqn:RLS}) the rigged Liouville space (RLS).
Proposition \ref{proposition2.1} also endows that the restriction of the unitary operator $I_C$ to $\Phi\hat{\otimes}\Phi$, $I_C|\Phi\hat{\otimes}\Phi$, becomes isomorphic from $(\Phi\hat{\otimes}\Phi,\tau_{\Phi\hat{\otimes}\Phi})$ onto $ (\Phi_{\mathcal{L}},\tau_{\Phi_\mathcal{L}})$. 
Note that whenever $\varphi,\psi\in\Phi$, $P_{\psi,\varphi}\in \Phi_{\mathcal{L}}$ holds by $\varphi\otimes C\psi\in \Phi\hat{\otimes}\Phi$.

By using the RLS (\ref{eqn:RLS}), bra and ket vectors can be constructed as follows: for $A\in \Phi_{\mathcal{L}}$, 
\begin{eqnarray}
\bra{A}_{\mathcal{L}} : \Phi_{\mathcal{L}}\to \mathbb{C},B\mapsto \langle A, B \rangle_{\mathcal{L}},
\label{eqn:bra_of_RLS}
\\
\ket{A}_{\mathcal{L}}:\Phi_{\mathcal{L}}\to \mathbb{C},B\mapsto \langle B, A \rangle_{\mathcal{L}}
\label{eqn:ket_of_RLS}.
\end{eqnarray}
These maps (\ref{eqn:bra_of_RLS}) and (\ref{eqn:ket_of_RLS}) are referred to as a super bra and a super ket, respectively.
Note that $\ket{A}_{\mathcal{L}}$ is the complex conjugate map of $\bra{A}_{\mathcal{L}}$, $\ket{A}_{\mathcal{L}}=\bra{A}_{\mathcal{L}}^*$, and $\bra{A}_{\mathcal{L}}\in \Phi_{\mathcal{L}}^\prime$ and $\ket{A}_{\mathcal{L}}\in \Phi_{\mathcal{L}}^\times$ hold.

Let us consider the relation between the super bra (\ref{eqn:bra_of_RLS}) (super ket (\ref{eqn:ket_of_RLS})) and the bra (ket) given in (\ref{eqn:bra_of_tensor_pro_of_RHS})-(\ref{eqn:ket_of_tensor_pro_of_RHS}).
Set the extension of $C$ to $\Phi^\prime\cup \Phi^\times$, $\hat{C} : \Phi^\prime\cup\Phi^\times\to\Phi^\prime\cup\Phi^\times$, by 
\begin{eqnarray}
(\hat{C} (f))(\phi)=f(C (\phi)),
    \label{eqn:C_extension}
\end{eqnarray}
for $f\in \Phi^\prime\cup\Phi^\times$ and $\phi\in \Phi$.
Then, $\hat{C}$ has the following properties.
(i) For $f\in\Phi^{\prime},~\hat{C}(f)\in \Phi^\times$, and for $f\in\Phi^{\times},~\hat{C}(f)\in \Phi^\prime$.
(ii) $\hat{C}:\Phi^{\prime}\to \Phi^{\times}$ and $\hat{C}:\Phi^{\times}\to \Phi^{\prime}$ are linear bijective (isomorphic).
(iii) For $\varphi\in \Phi$,
\begin{eqnarray}
\bra{\varphi}_{\mathcal{H}}\hat{C}\equiv\hat{C}(\bra{\varphi}_{\mathcal{H}})=\hat{C}(\ket{\varphi}^*_{\mathcal{H}})=\ket{C\varphi}_{\mathcal{H}},~~
\bra{\varphi}^*_{\mathcal{H}}\hat{C}\equiv\hat{C}(\bra{\varphi}_{\mathcal{H}}^*)=\hat{C}\ket{\varphi}_{\mathcal{H}}=\bra{C\varphi}_{\mathcal{H}}. 
    \label{eqn:C_relation}
\end{eqnarray}
The properties (i) and (ii) are shown by the anti-linearity and the bijection of $C$.
The property (iii) is derived from (i) and (ii).
We also define the extension of $I_C$, $\hat{I}_C:(\Phi\hat{\otimes}\Phi)^\prime\cup(\Phi\hat{\otimes}\Phi)^\times\to \Phi_{\mathcal{L}}^\prime \cup \Phi_{\mathcal{L}}^\times$, by 
\begin{eqnarray}
(\hat{I}_C (f))(A)=f(I_C^{-1} (A))
    \label{eqn:I_extension}
\end{eqnarray}
for $f\in\Phi\hat{\otimes}\Phi$ and $A \in \Phi_{\mathcal{L}}$.
Since $I_C$ is isomorphic,  $\hat{I}_C$ is also isomorphic from $(\Phi\hat{\otimes}\Phi)^j$ onto $\Phi_{\mathcal{L}}^j$, $(j=\prime,\times)$.
Note that its inverse $\hat{I}_C^{-1}$ satisfies $\hat{I}_C^{-1}=\widehat{I_C^{-1}}$ where $\widehat{I_C^{-1}}$ is the extension of the inverse $I_C^{-1}$ of $I_C$.
Using the extensions $\hat{C}$ and $\hat{I}_C$, for $\varphi,\psi\in \Phi$,
we have
\begin{eqnarray}
    \bra{P_{\psi,\varphi}}_{\mathcal{L}} & = & \bra{I_C(\varphi\otimes C\psi)}_{\mathcal{L}}
    =\bra{\varphi\otimes C\psi}_{\mathcal{H}\bar{\otimes}\mathcal{H}}\hat{I}_C
    \nonumber\\
    & = & 
    (\bra{\varphi}_{\mathcal{H}}\otimes \bra{C\psi}_{\mathcal{H}})\hat{I}_C
    = (\bra{\varphi}_{\mathcal{H}}\otimes \bra{\psi}^*_{\mathcal{H}}\hat{C})\hat{I}_C.
     \label{eqn:bra_calculation_1}
\end{eqnarray}
Similarly, 
\begin{eqnarray}
    \ket{P_{\psi,\varphi}}_{\mathcal{L}} & = & \ket{I_C(\varphi\otimes C\psi)}_{\mathcal{L}}
    =\hat{I}_C\ket{\varphi\otimes C\psi}_{\mathcal{H}\bar{\otimes}\mathcal{H}}
    \nonumber\\
    & = & 
    \hat{I}_C(\ket{\varphi}_{\mathcal{H}}\otimes \ket{C\psi}_{\mathcal{H}}) 
    = 
    \hat{I}_C(\ket{\varphi}_{\mathcal{H}}\otimes \hat{C}\ket{\psi}^*_{\mathcal{H}}).
     \label{eqn:bra_calculation_2}
\end{eqnarray}
Based on these relations,
the super bra and the super ket can be associated with the bra and ker of the tensor product of RHSs.

When denoting the super bra and super ket of $P_{\psi,\varphi}$ by
\begin{eqnarray}
\langle \! \langle{\varphi,\psi}|_\mathcal{L}\equiv\bra{P_{\psi,\varphi}}_{\mathcal{L}}\text{~~~and~~~}
| \varphi,\psi \rangle \! \rangle_\mathcal{L} \equiv\ket{P_{\psi,\varphi}}_{\mathcal{L}}
\label{eqn:double_bra_ket},
\end{eqnarray}
respectively, the relations
(\ref{eqn:bra_calculation_1}) and (\ref{eqn:bra_calculation_2}) can be expressed as 
\begin{eqnarray}
\langle \! \langle{\varphi,\psi}|_\mathcal{L}& =&(\bra{\varphi}_{\mathcal{H}}\otimes\bra{\psi}_{\mathcal{H}}^*\hat{C})\hat{I}_C,
\label{eqn:super_bra_relation_a}
\\
| \varphi,\psi \rangle \! \rangle_\mathcal{L} & = &\hat{I}_C(\ket{\varphi}_{\mathcal{H}}\otimes \hat{C}\ket{\psi}^*_{\mathcal{H}})
\label{eqn:super_bra_relation_b}.
\end{eqnarray}
Furthermore, because $\hat{C}$ is isomorphic, $\bra{\varphi}_\mathcal{H}^*\hat{C}$ and $\hat{C}\ket{\varphi}^*_\mathcal{H}$ can be identified with $\bra{\varphi}^*_\mathcal{H}$ and $\ket{\varphi}^*_\mathcal{H}$, respectively.
Through the identification, the relations (\ref{eqn:super_bra_relation_a}) and (\ref{eqn:super_bra_relation_b}) are represented by  
\begin{eqnarray}
\langle \! \langle{\varphi,\psi}|_\mathcal{L} =(\bra{\varphi}_{\mathcal{H}}\otimes\bra{\psi}_{\mathcal{H}}^*)\hat{I}_C \text{~~~and~~~}
| \varphi,\psi \rangle \! \rangle _\mathcal{L} = \hat{I}_C(\ket{\varphi}_{\mathcal{H}}\otimes\ket{\psi}^*_{\mathcal{H}}),
\label{eqn:super_bra_relation_identitification}
\end{eqnarray}
which can be often used in literatures concerning Liouville space formalism\cite{Gyamfi2020}.

\section{Spectral decomposition of Hermitian Liouville operator}
\label{sec:3}



In the previous section, we introduced the RLS (\ref{eqn:RLS}) by associating it with the tensor product of RHS (\ref{eqn:tensor_pro_of_RHS}) and provided the super bra (\ref{eqn:bra_of_RLS}) and the super ket (\ref{eqn:ket_of_RLS}) as the mappings belonging to $\Phi_{\mathcal{L}}^\prime$ and $\Phi_{\mathcal{L}}^\times$, respectively.
This section focuses on the Liouvillian operator (Liouvillian), denoted by $\mathcal{L}_H$, acting on the RLS and demonstrates how to describe its spectral expansions of the super bra and the super ket.
%
%
%
In general, to a Hamiltonian $H:D(H)\to\mathcal{H}$ on its domain $D$ into a Hilbert space $\mathcal{H}$, the corresponding Liouvillian operator $\mathcal{L}_H:D(\mathcal{L}_H)\to \mathcal{\mathcal{H}}$ can be constructed as a linear operator in the Liouville space $(\mathcal{L}(\mathcal{H}),\langle \cdot,\, \cdot\rangle_{\mathcal{L}})$, satisfying the unitary equivalence relation,
\begin{eqnarray}
    \mathcal{L}_H=I_C^{-1} L {I_C},
\label{eqn:relation_between_Liouvillian_and_L}
\end{eqnarray}
where $I_C$ is the unitary operator given by (\ref{eqn:LandT_unitary_relation}) and $L$ is the closure of the linear operator $H\otimes I -I\otimes CHC$ in $\mathcal{H}\bar{\otimes}\mathcal{H}$\cite{Spohn2004}, namely,
\begin{eqnarray}
    L=\overline{H\otimes I -I\otimes CHC} : {D}({L}) \to \mathcal{H} \bar{\otimes}\mathcal{H}. 
    \label{eqn:L_composite_operator}
\end{eqnarray}
Note that $I_C {D}({L})={D}({\mathcal{\mathcal{L}}_H})$.
When supposing that $H$ is self-adjoint, $L$ becomes self-adjoint and by (\ref{eqn:relation_between_Liouvillian_and_L}),  $\mathcal{L}$ is also self-adjoint.
The relation (\ref{eqn:relation_between_Liouvillian_and_L}) also provides that the spectral expansion for $\mathcal{L}_H$ can be systematically derived from the spectral expansion for $L$, which are obtained based on the tensor product of RHS (\ref{eqn:tensor_pro_of_RHS}).
In the following subsections, we first describe the spectral expansions for $L$.
And based on the obtained descriptions the spectral expansions of $\mathcal{L}_H$ are constructed in the succeeding sections.
Note that although the present section treats a Hermitian (self-adjoint) Liouville operator, the methodology developed here exhibits particular advantages when applied to a quasi-Hermitian Liouville operator, as will be discussed in Sec. \ref{sec:4}.

\subsection{Spectral expansion of $L$}
\label{subsec.3.1}

Assume that the Hamiltonian $H$ is continuous on $(\Phi,\tau_\Phi)$ and $H\Phi\subset \Phi$.  
Since $C$ is continuous anti-linear bijective from $\Phi$ onto $\Phi$, $CHC$ is continuous linear on $\Phi$ with the relation $CHC\Phi\subset\Phi$.
Thus, $L$, given by (\ref{eqn:L_composite_operator}), is continuous on $(\Phi\hat\otimes\Phi,\tau_{\Phi\hat\otimes\Phi})$ and satisfies $L(\Phi\hat\otimes\Phi)\subset \Phi\hat\otimes\Phi$.
Note that if $CD(H)\subset D(H)$ is satisfied, the relations $CD(H)=D(H)$ and $Sp(CHC)=Sp(H)$ hold.
When defining the extensions of $H$ and $CHC$, $\hat{H},~\widehat{CHC}:\Phi^\prime\cup \Phi^\times\to \Phi^\prime\cup \Phi^\times$, by $\widehat{H}(f)(\varphi)=f(H\varphi)$ and $\widehat{CHC}(f)(\varphi)=f(CHC\varphi)$ for $f\in \Phi^\prime\cup \Phi^\times,\varphi\in \Phi$, respectively,
it is easy to confirm the relation, $\widehat{CHC}=\hat{C}\hat{H}\hat{C}$, where $\hat{C}$ is the extension of $C$ given by (\ref{eqn:C_extension}).

From the nuclear spectral theorem\cite{Maurin1968}, $H$ has a set of generalized eigenvectors $\{\ket{\lambda}_{\mathcal{H}}\}_{\lambda \in Sp(H)}$, each $\ket{\lambda}_\mathcal{H}\in \Phi^\times$ corresponding to the eigenvalue $\lambda\in Sp(H)$ for $H$, such that the following eigenequations are satisfied,
\begin{eqnarray}
    \bra{\lambda}_\mathcal{H}\hat{H}=\lambda\bra{\lambda}_\mathcal{H},~~
        \hat{H}\ket{\lambda}_\mathcal{H}=\lambda\ket{\lambda}_\mathcal{H},
       \label{eqn:eigenequation_for_H}
\end{eqnarray}
where $\bra{\lambda}_\mathcal{H}\equiv\ket{\lambda}_\mathcal{H}^*\in \Phi^\prime$.
Furthermore, for any $\varphi,\psi\in \Phi$,
\begin{eqnarray}
    \langle \varphi, \psi\rangle_\mathcal{H} = \int_{Sp(H)} \braket{\varphi}{\lambda}_{\mathcal{H}}\braket{\lambda}{\psi}_{\mathcal{H}} d\mu (\lambda),
    \label{eqn:relation_3.2.1a}
    \\
    \langle \varphi, H\psi\rangle_\mathcal{H} = \int_{Sp(H)} \lambda\braket{\varphi}{\lambda}_{\mathcal{H}}\braket{\lambda}{\psi}_{\mathcal{H}} d\mu (\lambda),    \label{eqn:relation_3.2.1b}
\end{eqnarray}
where $\mu$ is the Borel measure. 
Since $\ket{\lambda}_\mathcal{H}$ is the generalized eigenvalue for $H$, $\ket{\lambda}_\mathcal{H}(H\varphi)=\lambda\ket{\lambda}_\mathcal{H}(\varphi)$ holds for $\varphi\in\Phi$.
Then, we have
\begin{eqnarray}
\hat{C}\ket{\lambda}_\mathcal{H}(CHC\varphi)=\ket{\lambda}_\mathcal{H}(CCHC\varphi)=\ket{\lambda}_\mathcal{H}(HC\varphi)
\nonumber \\
=\lambda\ket{\lambda}_\mathcal{H}(C\varphi)=\hat{C}(\lambda\ket{\lambda}_\mathcal{H})(\varphi)=\lambda\hat{C}\ket{\lambda}_\mathcal{H}(\varphi).
\label{eqn:generalized_eigenvector_for_C_calclulation}
\end{eqnarray}
Here $CC=C^2=I$ and the linearity of $\hat{C}$ are used.
Note that 
for $\ket{\lambda}_\mathcal{H}\in\Phi^\times$, $\hat{C}\ket{\lambda}_\mathcal{H}\in\Phi^\prime$, because $\hat{C}$ transforms an element of $\Phi^\times$ to $\Phi^\prime$. 
Similarly, $\bra{\lambda}_\mathcal{H}\hat{C}(\equiv\hat{C}(\bra{\lambda}_\mathcal{H}))\in\Phi^\times$ for $\bra{\lambda}_\mathcal{H}\in \Phi^\prime$.
(\ref{eqn:generalized_eigenvector_for_C_calclulation}) shows that $\hat{C}\ket{\lambda}_\mathcal{H}$ is a generalized eigenvector of $CHC$ corresponding to the eigenvalue $\lambda\in Sp(CHC)=Sp(H)$.
By introducing the notations, 
\begin{eqnarray}
\bra{\lambda^C}_\mathcal{H}\equiv\hat{C}\ket{\lambda}_\mathcal{H}=\hat{C}(\bra{\lambda}_\mathcal{H}^*),~~
\ket{\lambda^C}_\mathcal{H}\equiv\bra{\lambda}_\mathcal{H}\hat{C}=\hat{C}\ket{\lambda}_\mathcal{H}^*,
\label{eqn:notation_of_generalized_eigenvector_for_C}
\end{eqnarray}
the following eigenequations are obtained using (\ref{eqn:generalized_eigenvector_for_C_calclulation}).
\begin{eqnarray}
    \bra{\lambda^C}_\mathcal{H}\hat{C}\hat{H}\hat{C}=\lambda\bra{\lambda^C}_\mathcal{H},~~
        \hat{C}\hat{H}\hat{C}\ket{\lambda^C}_\mathcal{H}=\lambda\ket{\lambda^C}_\mathcal{H}.
       \label{eqn:eigenequation_for_CHC}
\end{eqnarray}
Note that $\bra{\lambda^C}_\mathcal{H}\in \Phi^\prime$ and $\ket{\lambda^C}_\mathcal{H}\in \Phi^\times$. 
Furthermore, the following relations are satisfied: for $\varphi,\psi \in \Phi$,
\begin{eqnarray}
    \langle \varphi, \psi\rangle_\mathcal{H} & = & \int_{Sp(H)} \braket{\varphi}{\lambda^C}_{\mathcal{H}}\braket{\lambda^C}{\psi}_{\mathcal{H}} d\mu (\lambda),
    \label{eqn:relation_3.2.2a}\\
       \langle \varphi, CHC\psi\rangle_\mathcal{H} & = & \int_{Sp(H)} \lambda \braket{\varphi}{\lambda^C}_{\mathcal{H}}\braket{\lambda^C}{\psi}_{\mathcal{H}} d\mu (\lambda),   \label{eqn:relation_3.2.2b}
\end{eqnarray}
where $\braket{\varphi}{\lambda^C}_\mathcal{H}=\ket{\lambda^C}_\mathcal{H}(\varphi)$ and $\braket{\lambda^C}{\varphi}_\mathcal{H}=\bra{\lambda^C}_\mathcal{H}(\varphi)$.
Actually, 
$\langle \varphi, \psi\rangle_\mathcal{H} = \langle C\psi, C\varphi\rangle_\mathcal{H} = \int_{Sp(H)} \braket{C\psi}{\lambda}_{\mathcal{H}}\braket{\lambda}{C\varphi}_{\mathcal{H}} d\mu (\lambda)$
where 
$\braket{C\psi}{\lambda}_{\mathcal{H}}=\ket{\lambda}_\mathcal{H}(C\psi)$ $=\hat{C}\ket{\lambda}_\mathcal{H}(\psi)=\bra{\lambda^C}_\mathcal{H}(\psi)=\braket{\lambda^C}{\psi}$ 
and 
$\braket{\lambda}{C\varphi}_{\mathcal{H}}=\bra{\lambda}_\mathcal{H}(C\varphi)=(\hat{C}\bra{\lambda}_\mathcal{H})(\varphi)=(\bra{\lambda}_\mathcal{H}\hat{C})(\varphi)=\ket{\lambda^C}_\mathcal{H}(\varphi)=\braket{\varphi}{\lambda^C}_\mathcal{H}$.
Thus, (\ref{eqn:relation_3.2.2a}) is obtained.
Using (\ref{eqn:relation_3.2.1b}) and $C^2=I$, we also obtain (\ref{eqn:relation_3.2.2b}).
%
%

For the sets of generalized eigenvectors $\{\ket{\lambda}_\mathcal{H}\}$ corresponding to $H$ and $\{\ket{\lambda^C}_\mathcal{H}\}$ corresponding to $CHC$,
$L$ has the set of the generalized eigenvectors $\{\ket{\lambda_1}_\mathcal{H}\otimes\hat{C}\ket{\lambda_2}_\mathcal{H}\}_{\lambda_1,\lambda_2\in Sp(H)}$ by which the following relations can be obtained: for any $\varphi, \psi \in \Phi \hat{\otimes} \Phi$, 
\begin{eqnarray}
    \langle \varphi,\, \psi\rangle_{\mathcal{H}{\bar{\otimes}} \mathcal{H}}
    & = & 
     \displaystyle
     \int_{\lambda\in Sp({L})}
     \braket{\hat{\varphi}}{\hat{\psi}}_\lambda
     d\nu_\lambda,
    \label{eqn:relation_3.2.3a}\\
    \langle \varphi,\, L\psi\rangle_{\mathcal{H}{\bar{\otimes}} \mathcal{H}}
    & = & 
     \displaystyle
     \int_{\lambda\in Sp({L})}
     \lambda\braket{\hat{\varphi}}{\hat{\psi}}_\lambda
     d\nu_\lambda,
    \label{eqn:relation_3.2.3b}
\end{eqnarray}
with
\begin{eqnarray}
    \displaystyle\braket{\hat{\varphi}}{\hat{\psi}}_\lambda
    =
    \displaystyle\int_{\lambda=\lambda_1-\lambda_2}
     (\ket{\lambda_1}_{\mathcal{H}}\otimes
     \ket{\lambda_2^C}_{\mathcal{H}})(\varphi)
     (\bra{\lambda_1}_{\mathcal{H}}\otimes
     \bra{\lambda_2^C}_{\mathcal{H}})(\psi)
     d\sigma^\lambda_{\lambda_1,\lambda_2},
    \label{eqn:relation_3.2.3c}
\end{eqnarray}
where $\nu_\lambda$ is the Borel measure obtained in the Neumanns' complete spectral theorem\cite{Maurin1968},  
$\sigma^\lambda_{\lambda_1,\lambda_2}$ is also a Borel measure on $\mathbb{R}^2$ whose support is contained in the set $\{(\lambda_1,\lambda_2)\in \mathbb{R}^2 ; \lambda=\lambda_1-\lambda_2, \lambda_i\in Sp(H), i=1,2\}$.
Note that $\bra{\lambda_1}_\mathcal{H}\otimes\bra{\lambda_2^C}_\mathcal{H}\in \Phi^{\prime}\otimes\Phi^\prime\subset(\Phi\hat{\otimes}\Phi)^\prime$ 
and $\ket{\lambda_1}_\mathcal{H}\otimes\ket{\lambda_2^C}_\mathcal{H}\in \Phi^{\times}\otimes\Phi^\times\subset(\Phi\hat{\otimes}\Phi)^\times$.
Here we adopt the following abbreviation \cite{Ohmori2024}:
\begin{eqnarray}
    \int_{\lambda\in Sp(L)}\int_{\lambda=\lambda_1-\lambda_2}\to \int_{Sp(L)}~~ 
    \text{and}~~
    d\sigma^\lambda_{\lambda_1,\lambda_2}d\nu_\lambda \to d\nu.
    \label{abbreviation}
\end{eqnarray}
Also, we denote $(\bra{\lambda_1}_{\mathcal{H}}\otimes
     \bra{\lambda_2^C}_{\mathcal{H}})(\varphi)$ and $(\ket{\lambda_1}_{\mathcal{H}}\otimes
     \ket{\lambda_2^C}_{\mathcal{H}})(\varphi)$ by 
     $\bra{\lambda_1}_{\mathcal{H}}\otimes\bra{\lambda_2^C}_{\mathcal{H}}\ket{\varphi}_{\mathcal{H}\bar{\otimes}\mathcal{H}}
     $ and $\bra{\varphi}_{\mathcal{H}\bar{\otimes}\mathcal{H}}\ket{\lambda_1}_\mathcal{H}\otimes \ket{\lambda_2^C}_\mathcal{H}$, respectively.
Using these expressions, the relations (\ref{eqn:relation_3.2.3a}) and (\ref{eqn:relation_3.2.3b}) can be converted to the following forms, simply:
\begin{eqnarray}
    \langle \varphi,\, \psi\rangle_{\mathcal{H}{\bar{\otimes}} \mathcal{H}}
     & = & 
    \displaystyle\int_{Sp(L)}\bra{\varphi}_{\mathcal{H}\bar{\otimes}\mathcal{H}}
    \ket{\lambda_1}_{\mathcal{H}}\otimes\ket{\lambda_2^C}_\mathcal{H}\bra{\lambda_1}_\mathcal{H}\otimes\bra{\lambda_2^C}_\mathcal{H}\ket{\psi}_{\mathcal{H}{\bar{\otimes}} \mathcal{H}}
     d\nu,
    \label{eqn:relation_3.2.4a}
    \\
    \langle \varphi,\, L\psi\rangle_{\mathcal{H}{\bar{\otimes}} \mathcal{H}}
    & = &
    \displaystyle\int_{Sp(L)}\lambda\bra{\varphi}_{\mathcal{H}\bar{\otimes}\mathcal{H}}
    \ket{\lambda_1}_{\mathcal{H}}\otimes\ket{\lambda_2^C}_\mathcal{H}\bra{\lambda_1}_\mathcal{H}\otimes\bra{\lambda_2^C}_\mathcal{H}\ket{\psi}_{\mathcal{H}{\bar{\otimes}} \mathcal{H}}
     d\nu.
      \label{eqn:relation_3.2.4b}
\end{eqnarray}
From them the spectral expansions of the bra and ket vectors for $L$ are obtained, as follows: for $\varphi \in \Phi_1 \hat{\otimes} \Phi_2$,
    \begin{eqnarray}
    \bra{\varphi}_{\mathcal{H}{\bar\otimes} \mathcal{H}}
    & = &
    \displaystyle\int_{Sp(L)}\bra{\varphi}_{\mathcal{H}\bar{\otimes}\mathcal{H}}
    \ket{\lambda_1}_{\mathcal{H}}\otimes\ket{\lambda_2^C}_\mathcal{H}\bra{\lambda_1}_\mathcal{H}\otimes\bra{\lambda_2^C}_\mathcal{H}
     d\nu,
    \label{eqn:spectral_expansion_bra_for_La}
     \\
    %
    \bra{L\varphi}_{\mathcal{H}{\bar\otimes} \mathcal{H}}
    & = &
    \displaystyle\int_{Sp(L)}\lambda\bra{\varphi}_{\mathcal{H}\bar{\otimes}\mathcal{H}}
    \ket{\lambda_1}_{\mathcal{H}}\otimes\ket{\lambda_2^C}_\mathcal{H}\bra{\lambda_1}_\mathcal{H}\otimes\bra{\lambda_2^C}_\mathcal{H}
     d\nu,
    \label{eqn:spectral_expansion_bra_for_Lb}
\end{eqnarray}
and 
\begin{eqnarray}
     \ket{\varphi}_{\mathcal{H}{\bar\otimes} \mathcal{H}}
    & = &
    \displaystyle\int_{Sp(L)}
\bra{\lambda_1}_\mathcal{H}\otimes\bra{\lambda_2^C}_\mathcal{H}\ket{\varphi}_{\mathcal{H}\bar{\otimes}\mathcal{H}}\ket{\lambda_1}_\mathcal{H}\otimes\ket{\lambda_2^C}_\mathcal{H}
    d\nu,
    \label{eqn:spectral_expansion_ket_for_La}
    \\
         \ket{L\varphi}_{\mathcal{H}{\bar\otimes} \mathcal{H}}
         & = & 
    \displaystyle\int_{Sp(L)}\lambda
     \bra{\lambda_1}_\mathcal{H}\otimes\bra{\lambda_2^C}_\mathcal{H}\ket{\varphi}_{\mathcal{H}\bar{\otimes}\mathcal{H}}\ket{\lambda_1}_\mathcal{H}\otimes\ket{\lambda_2^C}_\mathcal{H}.
    \label{eqn:spectral_expansion_ket_for_Lb}
    \end{eqnarray}
Note that $\{\ket{\lambda_1}_\mathcal{H}\otimes\ket{\lambda_2^C}_\mathcal{H}\}$ constitutes the following orthonormal base\cite{Maurin1968}:
\begin{eqnarray}
    I  = \displaystyle\int_{Sp(L)}\ket{\lambda_1}_\mathcal{H}\otimes\ket{\lambda_2^C}_\mathcal{H}\bra{\lambda_1}_\mathcal{H}\otimes\bra{\lambda_2^C}_\mathcal{H}     
    d\nu,
    \label{eqn:completion_for_L_a}
\end{eqnarray}
and
\begin{eqnarray}
    \bra{\lambda_1}_\mathcal{H}\otimes\bra{\lambda_2^C}_\mathcal{H}
    \ket{\lambda_1^\prime}_\mathcal{H}\otimes\ket{\lambda_2^{\prime C}}_\mathcal{H}=\check{\delta}(\lambda_1-\lambda_1^\prime)\check{\delta}(\lambda_2^{\prime C}-\lambda_2^{\prime C}),
    \label{eqn:completion_for_L_b}
\end{eqnarray}
where $\Check{\delta}$ is the delta-function specifying
\begin{eqnarray}
    f(\lambda_1^\prime,\lambda_2^\prime)
     = 
    \int_{Sp(A)}
    f(\lambda_1,\lambda_2)\Check{\delta}(\lambda_1^\prime-\lambda_1)
    \Check{\delta}(\lambda_2^\prime-\lambda_2)
    d\nu
    \label{deltafunction}
\end{eqnarray}
for any function $f(\lambda_1,\lambda_2)$.

\subsection{Extension of $L$}
\label{sec.4.2}

To derive the spectral expansions for $\mathcal{L}_H$, we next prepare the extension of $L$ to $(\Phi\hat{\otimes}\Phi)^\prime\cup (\Phi\hat{\otimes}\Phi)^\times$ that is defined by
\begin{eqnarray}
    (\hat {L} (f))(\varphi):=f(L(\varphi)),
    \label{eqn:extension_of_L}
\end{eqnarray}
for $f\in (\Phi\hat{\otimes}\Phi)^\prime\cup (\Phi\hat{\otimes}\Phi)^\times$ and $\varphi\in\Phi\hat{\otimes}\Phi$.
Then, it is confirmed that $\hat{L}$ satisfies the following eigenequations with respect to the generalized eigenvector $\ket{\lambda_1}_\mathcal{H}\otimes\ket{\lambda_2^C}$:  
\begin{eqnarray}
    \bra{\lambda_1}_\mathcal{H}\otimes\bra{\lambda_2^C}_\mathcal{H}\hat{L}=(\lambda_1-\lambda_2)\bra{\lambda_1}_\mathcal{H}\otimes\bra{\lambda_2^C}_\mathcal{H},
    \label{eqn:eigenequation_for_La}
    \\
        \hat{L}\ket{\lambda_1}_\mathcal{H}\otimes\ket{\lambda_2^C}_\mathcal{H}=(\lambda_1-\lambda_2)\ket{\lambda_1}_\mathcal{H}\otimes\ket{\lambda_2^C}_\mathcal{H}.
       \label{eqn:eigenequation_for_Lb}
\end{eqnarray}

Let us consider a representation of $\hat{L}$ using $\hat{H}$.
When putting a set $\check{\Phi}\equiv (\Phi^\prime\otimes\Phi^\prime)\cup(\Phi^\times\otimes\Phi^\times)$, $\check{\Phi}$ can be identified with the subset of $(\Phi\hat{\otimes}\Phi)^\prime\cup(\Phi\hat{\otimes}\Phi)^\times$.
Since $L$ is given as (\ref{eqn:L_composite_operator}), it is confirmed that the restriction of the extension $\hat{L}$ to $\check{\Phi}$ is expressed as  
\begin{eqnarray}
    \hat{L}|_{\check{\Phi}}=\hat{H}\otimes\hat{I}-\hat{I}\otimes\hat{C}\hat{H}\hat{C},
    \label{eqn:L_extension_relation_a}
\end{eqnarray}
where $\hat{I}:\Phi^j\to\Phi^j,f\mapsto f$~($j=\prime,\times$) is the identity map, 
$\hat{H}\otimes\hat{I}$ is the tensor product operator of $\hat{H}$ and $\hat{I}$ on $\check{\Phi}$, and $\hat{I}\otimes\hat{C}\hat{H}\hat{C}$ is also the tensor product operator of $\hat{I}$ and $\hat{C}\hat{H}\hat{C}=\widehat{CHC}$ on $\check{\Phi}$.
In general, this type of the operator (\ref{eqn:L_extension_relation_a}) can be extended to $(\Phi\hat{\otimes}\Phi)^\prime\cup (\Phi\hat{\otimes}\Phi)^\times$ by the following general procedure\cite{Ohmori2024}.
Set a pair of RHSs', $\Phi_i\subset\mathcal{H}_i\subset\Phi_i^{\prime}, \Phi_i^{\times}$ $(i=1,2)$. 
Suppose that $A:{D}(A)\to \mathcal{H}_1$ and $B:{D}(B)\to \mathcal{H}_2$ are  linear operators that satisfy continuity on $(\Phi_1,\tau_{\Phi_1})$ and $(\Phi_2,\tau_{\Phi_2})$ and $A\Phi_1\subset\Phi_1$ and $B\Phi_2\subset\Phi_2$, respectively.
For the extensions of $A$ and $B$, $\hat{A}:\Phi_1^\prime\cup\Phi_1^\times\to\Phi_1^\prime\cup\Phi_1^\times$ and 
$\hat{B}:\Phi_2^\prime\cup\Phi_2^\times\to \Phi_2^\prime\cup\Phi_2^\times$, let $\hat{A}\otimes \hat{B}$ be their tensor product operator from $\check{\Phi}$ into $\check{\Phi}$.
Then, the linear operator 
\begin{eqnarray}
\overline{\hat{A}\otimes\hat{B}}:(\Phi_1\hat\otimes\Phi_2)^\prime\cup (\Phi_1\hat\otimes\Phi_2)^\times\to (\Phi_1\hat\otimes\Phi_2)^\prime\cup (\Phi_1\hat\otimes\Phi_2)^\times,
\label{eqn:def_extension_dual_0}
\end{eqnarray}
defined by
\begin{eqnarray}
    \overline{\hat{A}\otimes\hat{B}}(f)(\phi)
    =
    f((A\otimes B) (\phi)),
\label{eqn:def_extension_dual}
\end{eqnarray}
for 
$ f\in (\Phi_1\hat\otimes\Phi_2)^\prime\cup (\Phi_1\hat\otimes\Phi_2)^\times,\phi\in\Phi_1\hat\otimes\Phi_2$, is 
the extension of $\hat{A}\otimes \hat{B}$ to $(\Phi\hat{\otimes}\Phi)^\prime\cup (\Phi\hat{\otimes}\Phi)^\times$.
Note that the extension is unique, and when setting the operator  
$\overline{\hat{A}_1\otimes\hat{B}_1\pm \hat{A}_2\otimes\hat{B}_2}$ on $(\Phi_1\hat\otimes\Phi_2)^\prime\cup (\Phi_1\hat\otimes\Phi_2)^\times$ by
$
\overline{\hat{A}_1\otimes\hat{B}_1\pm \hat{A}_2\otimes\hat{B}_2}(f)(\phi)
=
f((A_1\otimes B_1\pm A_2\otimes B_2) (\phi))$,
for 
$ f\in (\Phi_1\hat\otimes\Phi_2)^j ~(j=\prime,\times),\phi\in\Phi_1\hat\otimes\Phi_2$,
the relationship 
$\overline{\hat{A}_1\otimes\hat{B}_1\pm \hat{A}_2\otimes\hat{B}_2}=\overline{\hat{A}_1\otimes\hat{B}_1}\pm \overline{\hat{A}_2\otimes\hat{B}_2}$ is satisfied.
Therefore, by the definition (\ref{eqn:def_extension_dual}), 
the extension of $\hat{L}|_{\check{\Phi}}=\hat{H}\otimes\hat{I}-\hat{I}\otimes\hat{C}\hat{H}\hat{C}$ to $(\Phi\hat{\otimes}\Phi)^\prime\cup(\Phi\hat{\otimes}\Phi)^\times$ can be obtained as 
$\overline{\hat{H}\otimes\hat{I}-\hat{I}\otimes\hat{C}\hat{H}\hat{C}}$
and since such extension is unique, it coincides with $\hat{L}$.
Thus, we finally obtain the representation of $\hat{L}$ as
\begin{eqnarray}
    \hat{L}=\overline{\hat{H}\otimes\hat{I}-\hat{I}\otimes\hat{C}\hat{H}\hat{C}} : (\Phi\hat{\otimes}\Phi)^\prime\cup (\Phi\hat{\otimes}\Phi)^\times \to (\Phi\hat{\otimes}\Phi)^\prime\cup (\Phi\hat{\otimes}\Phi)^\times.
    \label{eqn:L_extension_relation_to_dual_spaces}
\end{eqnarray}

\subsection{Spectral expansion of $\mathcal{L}_H$}
\label{sec.4.3}

We now demonstrate the construction of the spectral expansions for the Liouvillian $\mathcal{L}_H$ based on the RLS (\ref{eqn:RLS}).
First, we extend the Liouvillian $\mathcal{L}_H$ to the dual spaces $\Phi^\prime_\mathcal{L}\cup\Phi^\times_\mathcal{L}$ in the same manner as $L$ in (\ref{eqn:extension_of_L}); its extension is denoted by $\hat{\mathcal{L}}_H$.
Using the isomorphism, $\hat{I}_C$, from $(\Phi \hat{\otimes}\Phi)^j$ onto $\Phi^j_\mathcal{L}$~$(j=\prime,\times)$ given as (\ref{eqn:I_extension}) in Sec. \ref{sec:2}, the relation between $\hat{\mathcal{L}}_H$ and $\hat{L}$ is shown as
\begin{eqnarray}
    \hat{\mathcal{L}}_H =\hat{I}_C^{-1}\hat{L}\hat{I}_C.
    \label{eqn:relation_L_and_Liouvillian}
\end{eqnarray}
For $\bra{\lambda_1}_\mathcal{H}\otimes\bra{\lambda_2^C}_\mathcal{H}\in (\Phi \hat{\otimes}\Phi)^\prime$ and $\ket{\lambda_1}_\mathcal{H}\otimes\ket{\lambda_2^C}_\mathcal{H}\in (\Phi \hat{\otimes}\Phi)^\times$, 
we introduce the notations, 
\begin{eqnarray}
\langle \! \langle{\lambda_1,\lambda_2}|_\mathcal{L}& \equiv & (\bra{\lambda_1}_{\mathcal{H}}\otimes\bra{\lambda_2^C}_{\mathcal{H}})\hat{I}_C,
\label{eqn:generalized_eigenvector_for_Liouvillian_a}\\
| \lambda_1,\lambda_2 \rangle \! \rangle_\mathcal{L} & \equiv & \hat{I}_C(\ket{\lambda_1}_{\mathcal{H}}\otimes \ket{\lambda_2^C}_{\mathcal{H}}).
\label{eqn:generalized_eigenvector_for_Liouvillian_b}
\end{eqnarray}
It is confirmed that they are generalized eigenvectors of $\mathcal{L}_H$ corresponding to $\lambda_1-\lambda_2$ in $Sp(\mathcal{L_H})=Sp(L)$ and the following eigenequations are satisfied:
\begin{eqnarray}
    \langle \! \langle{\lambda_1,\lambda_2}|_{\mathcal{L}}\hat{\mathcal{L}}_H =(\lambda_1-\lambda_2)\langle \! \langle{\lambda_1,\lambda_2}|_\mathcal{L},~~
    \hat{\mathcal{L}}_H| \lambda_1,\lambda_2 \rangle \! \rangle_{\mathcal{L}}=(\lambda_1-\lambda_2)| \lambda_1,\lambda_2 \rangle \! \rangle_\mathcal{L}.
    \label{eqn:generalized_eigenequations_for_Liouvillian}
\end{eqnarray}
Also, the following relations can be derived from (\ref{eqn:relation_3.2.4a}) and (\ref{eqn:relation_3.2.4b}): for $A,B\in \Phi_\mathcal{L}$,
    \begin{eqnarray}
    \langle A,\, B\rangle_{\mathcal{L}}
    & = &
    \displaystyle\int_{Sp(\mathcal{L}_H)}
   \langle A\, | \lambda_1,\lambda_2 \rangle \! \rangle_\mathcal{L} \langle \! \langle{\lambda_1,\lambda_2}|B\rangle_\mathcal{L}
     d\nu,\\
    \label{eqn:spectral_relation_for_Liouvillian_a}
    \langle A,\, \mathcal{L}_H B\rangle_{\mathcal{L}}
    & = &
    \displaystyle\int_{Sp(\mathcal{L}_H)}\lambda
   \langle A\, | \lambda_1,\lambda_2 \rangle \! \rangle _\mathcal{L}\langle \! \langle{\lambda_1,\lambda_2}|B\rangle_\mathcal{L}
     d\nu,
    \label{eqn:spectral_relation_for_Liouvillian_b}
    \end{eqnarray}
where $\langle A\, | \lambda_1,\lambda_2 \rangle \! \rangle_\mathcal{L}=| \lambda_1,\lambda_2 \rangle \! \rangle _\mathcal{L}(A)$ and $\langle \! \langle{\lambda_1,\lambda_2}|A\rangle_\mathcal{L}=\langle \! \langle{\lambda_1,\lambda_2}|_\mathcal{L}(A)$.
Therefore, the spectral expansions of the super bra $\bra{A}_\mathcal{L}$ and the super ket $\ket{A}_\mathcal{L}$ for the Liouvillian $\mathcal{L}_H$ by the set of its generalized eigenvectors $\{| \lambda_1,\lambda_2 \rangle \! \rangle_\mathcal{L}\}$ are obtained, as follows: for $A\in\Phi_\mathcal{L}$,
\begin{eqnarray}
    \bra{A}_{\mathcal{L}} & = & \int_{Sp(\mathcal{L}_H)} \langle A\, | \lambda_1,\lambda_2 \rangle \! \rangle_\mathcal{L} \langle \! \langle{\lambda_1,\lambda_2}|_\mathcal{L}
     d\nu,\\
     \label{eqn:expansions_of_bra_by_Liouvillian_a}
    \bra{A}_{\mathcal{L}}\hat{\mathcal{L}}_H & = & \int_{Sp(\mathcal{L}_H)} \lambda \langle A\, | \lambda_1,\lambda_2 \rangle \! \rangle _\mathcal{L}\langle \! \langle{\lambda_1,\lambda_2}|_\mathcal{L}
     d\nu,
    \label{eqn:expansions_of_bra_by_Liouvillian_b}
\end{eqnarray}
and
\begin{eqnarray}
    \ket{A}_{\mathcal{L}} & = & \int_{Sp(\mathcal{L}_H)} 
    \langle \! \langle{\lambda_1,\lambda_2}|A \rangle_\mathcal{L}
    | \lambda_1,\lambda_2 \rangle \! \rangle _\mathcal{L}
     d\nu,\\
     \label{eqn:expansions_of_ket_by_Liouvillian_a}
    \hat{\mathcal{L}}_H\ket{A}_{\mathcal{L}} & = & \int_{Sp(\mathcal{L}_H)} \lambda
    \langle \! \langle{\lambda_1,\lambda_2}|A \rangle
    _\mathcal{L}| \lambda_1,\lambda_2 \rangle \! \rangle _\mathcal{L}
     d\nu.
    \label{eqn:expansions_of_ket_by_Liouvillian_b}
\end{eqnarray}
Note that $\{| \lambda_1,\lambda_2 \rangle \! \rangle\}$ gives the orthonormal base that satisfies  
\begin{eqnarray}
    I =  \displaystyle\int_{Sp(\mathcal{L}_H)}
    | \lambda_1,\lambda_2 \rangle \! \rangle_\mathcal{L}
    \langle \! \langle \lambda_1,\lambda_2 |_\mathcal{L}
    d\nu,~~
\langle \! \langle \lambda_1,\lambda_2 |_\mathcal{L}| \lambda_1^\prime,\lambda_2^\prime \rangle \! \rangle_\mathcal{L}=\delta(\lambda_1-\lambda_1^\prime)\delta(\lambda_2-\lambda_2^\prime),
    \label{eqn:completion_for_Liouvillian}
\end{eqnarray}
as well as the orthonormal base, (\ref{eqn:completion_for_L_a}) and (\ref{eqn:completion_for_L_b}), for $\{\ket{\lambda_1}_\mathcal{H}\otimes\ket{\lambda_2^C}_\mathcal{H}\}$.

For $\varphi,\psi\in \Phi$, the super bra, $\langle \! \langle{\varphi,\psi}|_\mathcal{L}\in \Phi_{\mathcal{L}}^\prime$, and the super ket, $| \varphi,\psi \rangle \! \rangle_\mathcal{L}\in \Phi_{\mathcal{L}}^\times$, are given by (\ref{eqn:double_bra_ket}).
Owing to the reason, we here denote $\langle \! \langle \lambda_1,\lambda_2 |_\mathcal{L}(P_{\psi,\varphi})$ by $\langle \! \langle \lambda_1,\lambda_2 |\varphi,\psi \rangle \! \rangle _\mathcal{L}$ and $|\lambda_1,\lambda_2 \rangle \! \rangle_\mathcal{L} (P_{\psi,\varphi})$ by $\langle \! \langle \varphi,\psi |\lambda_1,\lambda_2 \rangle \! \rangle_\mathcal{L} $, respectively.
Then, the simple calculation,
$\langle \! \langle \lambda_1,\lambda_2 |_\mathcal{L}(P_{\psi,\varphi})=(\bra{\lambda_1}_\mathcal{H}\otimes\bra{\lambda_2^C}_\mathcal{H})\hat{I}_C(P_{\psi,\varphi})=\bra{\lambda_1}_\mathcal{H}\otimes\bra{\lambda_2^C}_\mathcal{H}(I^{-1}_CP_{\psi,\varphi})=\bra{\lambda_1}_\mathcal{H}\otimes\bra{\lambda_2^C}_\mathcal{H}(\varphi\otimes C\psi)=\braket{\lambda_1}{\varphi}_\mathcal{H}\bra{\lambda_2^C}_\mathcal{H}(C\psi)=\braket{\lambda_1}{\varphi}_\mathcal{H}\hat{C}\ket{\lambda_2}_\mathcal{H}(C\psi)=\braket{\lambda_1}{\varphi}_\mathcal{H}\ket{\lambda_2}_{\mathcal{H}}(\psi)=\braket{\lambda_1}{\varphi}_{\mathcal{H}}\braket{\psi}{\lambda_2}_\mathcal{H}$,
leads to 
\begin{eqnarray}
    \langle \! \langle \lambda_1,\lambda_2 |\varphi,\psi \rangle \! \rangle_\mathcal{L} 
    & = &
    \braket{\lambda_1}{\varphi}_{\mathcal{H}}\braket{\psi}{\lambda_2}_\mathcal{H}  =  \braket{\lambda_1}{\varphi}\braket{\lambda_2}{\psi}^*,
    \label{eqn:relation_for_Liouvillian_braket_a}
    \\
\langle \! \langle \varphi,\psi |\lambda_1,\lambda_2 \rangle \! \rangle_\mathcal{L}
    & = &
    \braket{\varphi}{\lambda_1}_{\mathcal{H}}\braket{\lambda_2}{\psi}_\mathcal{H}
    =\braket{\varphi}{\lambda_1}_\mathcal{H}\braket{\psi}{\lambda_2}^*.
    \label{eqn:relation_for_Liouvillian_braket_b}
\end{eqnarray}
Using these relations, the spectral expansions for $\langle \! \langle{\varphi,\psi}|_\mathcal{L}$ and $| \varphi,\psi \rangle \! \rangle_\mathcal{L}$ by $\{| \lambda_1,\lambda_2 \rangle \! \rangle_\mathcal{L}\}$ can be described, as follows: for $\varphi,\psi\in \Phi$,
    \begin{eqnarray}
    \langle \! \langle \varphi,\psi |_\mathcal{L}
    & = &
    \displaystyle\int_{Sp(\mathcal{L}_H)}
    \braket{\varphi}{\lambda_1}_{\mathcal{H}}\braket{\psi}{\lambda_2}_\mathcal{H}^*\langle \! \langle \lambda_1,\lambda_2 |
     d\nu,
    \label{eqn:spectral_expansion_double_bra_for_Lipivillian_a}
    \\
    %
        | \varphi,\psi \rangle \! \rangle_\mathcal{L}
    & = &
    \displaystyle\int_{Sp(\mathcal{L}_H)}
    \braket{\lambda_1}{\varphi}_{\mathcal{H}}\braket{\lambda_2}{\psi}_\mathcal{H}^*
    | \lambda_1,\lambda_2 \rangle \! \rangle
     d\nu.
    \label{eqn:spectral_expansion_double_ket_for_Lipivillian_b}
    \end{eqnarray}
These relations (\ref{eqn:spectral_expansion_double_bra_for_Lipivillian_a}) and (\ref{eqn:spectral_expansion_double_ket_for_Lipivillian_b}) indicate the spectral expansions for the Liouvillian $\mathcal{L}_H$ associated with the generalized eigenvector $\{\ket{\lambda}\}_{\lambda\in Sp(H)}$ of the Hamiltonian $H$.

\subsection{Application to Harmonic Oscillator}
\label{sec:3.4}
%

Our RLS formulation shown in the present section can be applied to several Hermitian quantum systems.
For instance, we consider a one-dimensional harmonic oscillator system whose Hamiltonian is given by 
\begin{eqnarray}
    H_{ho}=\frac{1}{2\mu}P^2+\frac{\mu\omega ^2}{2}Q^2,
    \label{eqn:Harmonic_osc_Hamiltonian}
\end{eqnarray}
where $Q$ and $P$ are the position and momentum operators that satisfy the commutation relation $[P,Q]\equiv PQ-QP=-i\hbar I$ and $\mu,\omega$ are a real constant (mass and frequency, respectively).
It is well-known\cite{Bohm1978} that the RHS corresponding to this Hamiltonian can be established using the Hilbert space $\mathcal{H}_{ho}=(\mathcal{H}_{ho},\langle \cdot,\, \cdot\rangle_{ho})$ with the conuntable orthonormal base $\{\phi_n\}_{n=0}^\infty$ and its dense nuclear subspace $\Phi_{ho}=(\Phi_{ho},\tau_{\Phi_{ho}})$,
\begin{eqnarray}
    \Phi_{ho} \subset \mathcal{H}_{ho} \subset \Phi_{ho}^\prime,\Phi_{ho}^\times.
    \label{eqn:RHS_for_Harmonic_oscillator}
\end{eqnarray}
The nuclear topology $\tau_{\Phi_{ho}}$ is provided by the completion of the induced topology from the countable compatible inner products, $\langle\varphi,\, \psi\rangle_p\equiv\langle\varphi,\, (N+I)^p\psi\rangle_{ho},( p=0,1,2,\cdots)$, for $\varphi, \psi \in <\{\phi_n\}>$, where $I$ is the identity, $N$ is the number operator given as $N=\frac{1}{\omega \hbar}H_{ho}-\frac{1}{2}I$, and the notation $< A >$ shows the linear subspace spanned by a subset $A$ of $\mathcal{H}_{ho}$.
Note that $D(H_{ho})=D(N)=\{\sum_{n=0}^\infty\alpha_n\phi_n \mid \sum_{n=0}^\infty n^2|\alpha_n|<+\infty \}$ and $H_{ho}\phi_n=\hbar\omega(n+\frac{1}{2})\phi_n$ $(N\phi_n=n\phi_n)$ for $n=0,1,2,\cdots$.
In the RHS, $\Phi_{ho}\subset \mathcal{H}_{ho} \subset \Phi_{ho}^{\prime},\Phi_{ho}^\times$, the Hamiltonian $H_{ho}:D(H_{ho})\to \mathcal{H}_{ho}$ given by  (\ref{eqn:Harmonic_osc_Hamiltonian}) is continuous on $\Phi_{ho}$ and $H_{ho}\Phi_{ho} \subset \Phi_{ho}$ is satisfied.
We now set a map $C_{ho}:\mathcal{H}_{ho}\to \mathcal{H}_{ho}$ where $C_{ho}(\varphi)=\sum_{n=0}^\infty \alpha_n^*\phi_n$ for $\varphi=\sum_{n=0}^\infty \alpha_n\phi_n$.
It is confirmed that $C_{ho}$ is anti-linear, $C_{ho}^2=I$, and $\|C_{ho}\varphi\|_{ho}=\|\varphi\|_{ho}$, which shows that $C_{ho}$ is a conjugation.
In particular, $C_{ho}(\phi_n)=\phi_n$.
Furthermore, it can be easily verified that $C_{ho}(\Phi_{ho})=\Phi_{ho}, C_{ho}D(H_{ho})\subset D(H_{ho}),C_{ho}H_{ho}\subset H_{ho}C_{ho}$, and $C_{ho}$ is continuous from $(\Phi_{ho},\tau_{\Phi_{ho}})$ onto $(\Phi_{ho},\tau_{\Phi_{ho}})$.
The last property can be shown since $\tau_{\Phi_{ho}}$ is induced from the countable norms $\| \cdot\|_{p}=\sqrt{\langle\cdot,\, \cdot\rangle_p}$ and $C_{ho}(N+I)=(N+I)C_{ho}$ holds on $D(N)$.
Therefore, $C_{ho}$ satisfies the assumptions that we impose so far, and the general formulation developed here can be applied to describing the spectral expansions for the Liouville operator $\mathcal{L}_{H_{ho}}$ corresponding to ${H}_{ho}$. 

As $C_{ho}H_{ho}\subset H_{ho}C_{ho}$, the operator $L_{ho}$ can be represented as $L_{ho}=\overline{H_{ho}\otimes I-I\otimes H_{ho}}$ instead of (\ref{eqn:L_composite_operator}). 
For each generalized eigenvector $\ket{n}\equiv \ket{\phi_n}_{ho}$ of $H_{ho}$, $\hat{C}_{ho}\ket{n}=\bra{n}$ and $\bra{n}\hat{C}_{ho}=\ket{n}$.
Note that from (\ref{eqn:spectral_expansion_bra_for_La})-(\ref{eqn:spectral_expansion_ket_for_Lb}), the spectral expansions for $L_{ho}$ can be executed using $\{\ket{m}\otimes \ket{n}\}$ based on the RHS, $\Phi_{ho}\hat{\otimes}\Phi_{ho}\subset\mathcal{H}_{ho}\bar\otimes \mathcal{H}_{ho}\subset (\Phi_{ho}\hat{\otimes}\Phi_{ho})^\prime,(\Phi_{ho}\hat{\otimes}\Phi_{ho})^\times$.
Applying the representation (\ref{eqn:L_extension_relation_to_dual_spaces}) to the present case, we obtain 
\begin{eqnarray}
\hat{L}_{ho}=\overline{\hat{H}_{ho}\otimes\hat{I}-\hat{I}\otimes\hat{H}_{ho}},
\label{eqn:HO_L_extension_relation_to_dual_spaces}
\end{eqnarray}
and the relation (\ref{eqn:relation_L_and_Liouvillian}) shows 
\begin{eqnarray}
    \hat{\mathcal{L}}_{H_{ho}} =\hat{I}_{C_{ho}}^{-1}\hat{L}_{ho}\hat{I}_{C_{ho}}.
    \label{eqn:HO_relation_L_and_Liouvillian}
\end{eqnarray}
We now set the RLS, 
\begin{eqnarray}
    {\Phi}_{\mathcal{L}_{ho}} \subset \mathcal{L}(\mathcal{H}_{ho}) \subset 
    {\Phi}_{\mathcal{L}_{ho}}^{\prime},{\Phi}_{\mathcal{L}_{ho}}^{\times},
    \label{eqn:RLS_for_Harmonic_oscillator}
\end{eqnarray}
    given by (\ref{eqn:RLS}), corresponding to the RHS (\ref{eqn:RHS_for_Harmonic_oscillator}).
Noting $Sp(L_{ho})=Sp(\mathcal{L}_{\mathcal{H}_{ho}})=\{\hbar\omega (m-n)\mid m,n=0,1,2,\cdots\}$, the spectral expansions of the super bra and the super ket based on the RLS are expressed as follows:
for $A\in{\Phi}_{\mathcal{L}_{ho}}$,
\begin{eqnarray}
    \bra{A}_{\mathcal{L}_{ho}} & = & \sum_{m,n=0}^\infty \langle A\, | m,n \rangle \! \rangle_\mathcal{L} \langle \! \langle{m,n}|_\mathcal{L},
\label{eqn:HO_expansions_of_bra_by_Liouvillian_a}
    \\  
    \bra{A}_{\mathcal{L}_{ho}}\hat{\mathcal{L}}_{H_{ho}} & = & \sum_{m,n=0}^\infty \hbar \omega (m-n) \langle A\, | m,n \rangle \! \rangle_\mathcal{L} \langle \! \langle{m,n}|_\mathcal{L},\label{eqn:HO_expansions_of_bra_by_Liouvillian_b}
\end{eqnarray}
and
\begin{eqnarray}
    \ket{A}_{\mathcal{L}_{ho}} & = & \sum_{m,n=0}^\infty 
    \langle \! \langle{m,n}|A \rangle_\mathcal{L}
    | m,n \rangle \! \rangle _\mathcal{L}
     ,\\
     \label{eqn:HO_expansions_of_ket_by_Liouvillian_a}
    \hat{\mathcal{L}}_{H_{ho}}\ket{A}_{\mathcal{L}_{ho}} & = & \sum_{m,n=0}^\infty \hbar \omega (m-n) 
    \langle \! \langle{m,n}|A \rangle_\mathcal{L}
    | m,n \rangle \! \rangle_\mathcal{L},
\label{eqn:HO_expansions_of_ket_by_Liouvillian_b}
\end{eqnarray}
where $\langle \! \langle{m,n}|_\mathcal{L} =  \bra{m}\otimes\bra{n}\hat{I}_{C_{ho}}\in \Phi_{\mathcal{L}_{ho}}^\prime$ and $| m,n \rangle \! \rangle_\mathcal{L} = \hat{I}_{C_{ho}}\ket{m}\otimes \ket{n}\in \Phi_{\mathcal{L}_{ho}}^\times$.
Also the following eigenequations for $\hat{\mathcal{L}}_{H_{ho}}$ are obtained:
\begin{eqnarray}
    \langle \! \langle{m,n}|_\mathcal{L}\hat{\mathcal{L}}_{H_{ho}} =\hbar \omega(m-n)\langle \! \langle{m,n}|_\mathcal{L},~~
    \hat{\mathcal{L}}_{H_{ho}}| m,n \rangle \! \rangle_\mathcal{L}=\hbar \omega(m-n)| m,n \rangle \! \rangle_\mathcal{L}.
    \label{eqn:HO_generalized_eigenequations_for_Liouvillian}
\end{eqnarray}
Note that they are executed on $\Phi_{\mathcal{L}_{ho}}^\prime\cup\Phi_{\mathcal{L}_{ho}}^\times$.

\section{Spectral decomposition of Quasi-Hermitian Liouville operator}
\label{sec:4}

In the preceding section, we have dealt with the Hermitian Liouville operator related to a Hermitian Hamiltonian.
The present section targets a non-Hermitian case, especially, quasi-Hermitian Liouville operator and presents its RLS formalism.
Hereafter, for a given Hilbert space $\mathcal{H}=(\mathcal{H},\langle \cdot, \cdot\rangle_\mathcal{H})$ and a positive invertible operator $\eta$ on it, we denote $\mathcal{H}_\eta$ as the Hilbert space obtained by completing the pre-Hilbert space $\mathcal{H}_\eta=(\mathcal{H},\langle \cdot, \cdot \rangle_\eta)$ possessing the metric induced from the inner product $\langle \phi, \psi \rangle_\eta
    = \langle \phi, \eta \psi \rangle_\mathcal{H}$ for $\phi,\psi \in \mathcal{H}$ with respect to $\eta$\cite{Ohmori2022}. 


\subsection{Quasi Hermitian property of Liouville operators $L$ and $\mathcal{L}_H$}
\label{sec:4.1}

Let $\eta$ be a positive invertible operator on a Hilbert space $\mathcal{H}$ and $H:D(H)\to \mathcal{H}$ a densely defined closable operator.
Assume that $H$ is an $\eta$-quasi Hermitian operator, having the symmetric relation,
\begin{eqnarray}
    H^\dagger = \eta H \eta ^{-1}.
    \label{eqn:eta_quasi_Hermitian_H}
\end{eqnarray}
Here $\eta$ is regarded as the intertwining operator for $H$ and its adjoint $H^\dagger$ where $\eta D(H)=D(H^\dagger)$ holds.
Then, the operator $L$ of the form (\ref{eqn:L_composite_operator}) becomes an $\eta\otimes \eta$-quasi Hermitian operator.
To show it, we prepare the following two lemmas.
\begin{lemma}
\label{lemma4.1}
Let $A:D(A)\to \mathcal{H}$ be closable and let $C:\mathcal{H}\to \mathcal{H}$ be a conjugation such that $CD(A)\subset D(A)$ holds.
Then, $CAC:D(CAC)\to \mathcal{H}$ is closable.
\end{lemma}

\begin{proof}
From the assumptions, $CD(A)=D(A)$ and $D(CAC)=D(A)$.
For any sequence $\{\psi_n\}\subset D(CAC)$ with $\psi_n\to 0$ and $CAC\psi_n\to \phi$ as $n\to +\infty$, we show $\phi=0$.
By $\psi_n\in D(CAC)=D(A)$, $C\psi_n\in CD(A)=D(A)$ for each $n$.
Since $C$ is a conjugation, $\|C\psi_n\|_{\mathcal{H}}=\|\psi_n\|_{\mathcal{H}}$, by which $\psi_n\to 0$ implies $0=\lim_{n\to \infty}\|\psi_n \|_{\mathcal{H}}=\lim_{n\to \infty}\|C\psi_n \|_{\mathcal{H}}$.
Thus, $\{C\psi_n\}\subset D(A)$ satisfies $C\psi_n\to 0$.
From $0=\lim_{n\to \infty}\|CAC\psi_n-\phi \|_{\mathcal{H}}=\lim_{n\to \infty}\|C(AC\psi_n-C\phi) \|_{\mathcal{H}}=\lim_{n\to \infty}\|A(C\psi_n)-C\phi \|_{\mathcal{H}}$, we have $A(C\psi_n)\to C\phi$.
Since $A$ is closable, therefore, $C\phi=0$, which leads to $\phi=0$ because $C$ is injective.
\end{proof}

\begin{lemma}
\label{lemma4.2}
Let $A:D(A)\to \mathcal{H}$ be an $\eta$-quasi Hermitian, satisfying $A^\dagger=\eta A \eta ^{-1}$, where $\eta:\mathcal{H}\to \mathcal{H}$ is a positive intertwining operator.
Let $C:\mathcal{H}\to \mathcal{H}$ be a conjugation satisfying $CD(A)\subset D(A)$ and $CD(A^\dagger)\subset D(A^\dagger)$.
Furthermore, assume that $\eta$ and $C$ are commutative, namely
\begin{eqnarray}
    C\eta-\eta C=0~~\text{on $\mathcal{H}$}.
    \label{eqn:commutative_for_C_eta}
\end{eqnarray}
Then, $CAC : D(CAC)\to \mathcal{H}$ is an $\eta$-quasi Hermitian and satisfies the relation
\begin{eqnarray}
    CA^\dagger C=(CAC)^\dagger=\eta C A C \eta^{-1}.
    \label{eqn:CAC_quasi_Hermitian}
\end{eqnarray}
\end{lemma}

\begin{proof}
We first show $CA^\dagger C=(CAC)^\dagger$.
In the hypothesis, $CD(A)=D(A)=D(CAC)$ and $CD(A^\dagger)=D(A^\dagger)=D(CA^\dagger C)$ hold.
Let $\varphi\in D(A^\dagger)$ and $\psi\in D(A)$.
Then, $\langle \varphi, CAC\psi\rangle_\mathcal{H} = \langle AC\psi, C\varphi\rangle_\mathcal{H}=\langle C\psi, A^\dagger C\varphi\rangle_\mathcal{H}=\langle CA^\dagger C\varphi, \psi\rangle_\mathcal{H}$.
Hence, we obtain $CA^\dagger C\subset(CAC)^\dagger$.
Whereas, replacing $A$ with $CAC$ provides $D((CAC)^\dagger)\subset D(A^\dagger)$ by $C^2=I$.
Therefore, $(CAC)^\dagger\subset CA^\dagger C$.
Next, for any $\varphi\in D(A)$, we have $\eta CAC \varphi=C\eta AC\varphi=CA^\dagger \eta C\varphi=CA^\dagger C\eta \varphi$, which shows that $CAC$ is an $\eta$-quasi Hermitian.
\end{proof}

In the previous study\cite{Ohmori2024}, we confirmed that for any pair of closable $\eta_i$-quasi Hermitian operators $A_i$ in $\mathcal{H}_i$ where $\eta_i$ is a positive operator on $\mathcal{H}_i$ ($i=1,2$), the linear operator in the tensor product $\mathcal{H}_1\bar{\otimes}\mathcal{H}_2$, $A=\overline{A_1\otimes I+I\otimes A_2}$, is $\eta_1\otimes\eta_2$-quasi Hermitian operator and hence self-adjoint in $(\mathcal{H}_1\bar{\otimes}\mathcal{H}_2)_{\eta_1\otimes\eta_2}$.
Here $A$ is the closure of the operator $A_1\otimes I +I\otimes A_2$ in $\mathcal{H}_1 \bar{\otimes}\mathcal{H}_2$.
In the present case, for an $\eta$-quasi Hermitian Hamiltonian $H:D(H)\to \mathcal{H}$, the previous two lemmas show that $CHC:D(CHC)\to \mathcal{H}$ becomes a closable $\eta$-quasi Hermitian operator.
Note that $-CHC:D(CHC)\to \mathcal{H}$ is also a closable $\eta$-quasi Hermitian operator.
Therefore, when setting $A_1=H$, $A_2=-CHC$, $\eta_1=\eta_2=\eta$, and $\mathcal{H}_1=\mathcal{H}_2=\mathcal{H}$, the linear operator $A=\overline{H\otimes I+I\otimes (-CHC)}=\overline{H\otimes I-I\otimes CHC}$ is $\eta\otimes\eta$-quasi Hermitian operator in the tensor product $\mathcal{H}\bar{\otimes}\mathcal{H}$.
Consequently, we obtain the following result.
\begin{theorem}
\label{theorem4.1}
Let $H:D(H)\to \mathcal{H}$ be an $\eta$-quasi Hermitian Hamiltonian.
Then the linear operator $L$ of the form (\ref{eqn:L_composite_operator}) is the $\eta\otimes \eta$-quasi Hermitian in the tensor product $\mathcal{H}\bar{\otimes}\mathcal{H}$, having the symmetrical relation
\begin{eqnarray}
    L^\dagger = (\eta \otimes \eta) L (\eta\otimes \eta) ^{-1},
    \label{eqn:eta_quasi_Hermitian_L}
\end{eqnarray}
and hence $L$ is a self-adjoint in $(\mathcal{H} \bar{\otimes}\mathcal{H})_{\eta\otimes\eta}$.
\end{theorem}
%


We now set the unitary operator $I_C :  (\mathcal{H} \bar{\otimes} \mathcal{H},\langle \cdot,\, \cdot\rangle_{\mathcal{H} \bar{\otimes} \mathcal{H}}) \to (\mathcal{L}(\mathcal{H}),\langle \cdot,\, \cdot\rangle_{\mathcal{L}})$ that satisfies (\ref{eqn:LandT_unitary_relation}).
By $I_C$, the unitary transformation of $\eta\otimes \eta$  is defined as 
\begin{eqnarray}
    \zeta = I_C (\eta \otimes \eta) I_C^{-1}.
    \label{eqn:zeta}
\end{eqnarray}
This operator $\zeta:(\mathcal{L}(\mathcal{H}),\langle \cdot,\, \cdot\rangle_{\mathcal{L}})\to (\mathcal{L}(\mathcal{H}),\langle \cdot,\, \cdot\rangle_{\mathcal{L}})$ is  positive invertible because $\eta\otimes \eta$ has these properties.
Furthermore, the following theorem concerning the quasi Hermitian property for the Liouville operator is shown.
\begin{theorem}
\label{theorem4.2}
Let $H$ be an $\eta$-quasi Hermitian Hamiltonian and let $L=\overline{H\otimes I-I\otimes CHC}$ given as (\ref{eqn:L_composite_operator}).
Define the Liouville operator for $H$ by
\begin{eqnarray}
    \mathcal{L}_H = I_C L I_C^{-1}:D(\mathcal{L}_H) \to \mathcal{L}(\mathcal{H}).
    \label{eqn:quasi_Liouvillian}
\end{eqnarray}
Then, $\mathcal{L}_H$ is a $\zeta$-quasi Hermitian operator, exhibiting the symmetric relation: 
\begin{eqnarray}
    \mathcal{L}_H^\dagger = \zeta \mathcal{L}_H\zeta^{-1}.
    \label{eqn:zeta_quasi_Liouvillian}
\end{eqnarray}
\end{theorem}
\begin{proof}
Using (\ref{eqn:eta_quasi_Hermitian_L}), let us directly show the relation (\ref{eqn:zeta_quasi_Liouvillian}), as follows.
\begin{eqnarray}
    \mathcal{L}_H^\dagger & = & (I_C^{-1})^\dagger L^\dagger I_C^\dagger=I_CL^\dagger I_C^{-1}=I_C(\eta\otimes \eta) L (\eta\otimes \eta)^{-1}I_C^{-1}
    \nonumber
    \\
    & = & (I_C \eta\otimes \eta I_C^{-1}) \mathcal{L}_H (I_C\eta\otimes \eta I_C^{-1})^{-1}=\zeta \mathcal{L}_H \zeta^{-1}.
    \label{eqn:calculation_thm4.2}
\end{eqnarray}
In addition, we have $\zeta D(\mathcal{L}_H)=I_C(\eta \otimes \eta)D(L)=I_CD(L^\dagger)=D(\mathcal{L}_H^\dagger)$.
Therefore, $\mathcal{L}_H$ is a $\zeta$-quasi Hermitian operator.
\end{proof}
\noindent
Note that both $L$ and $\mathcal{L}_H$ are {\it not} self-adjoint on $\mathcal{H}\bar{\otimes}\mathcal{H}$ and $\mathcal{L}_H$, respectively, since  $H$ is not self-adjoint.

\subsection{$\zeta$-RHS}
\label{sec:4.3}

For a $\eta$-quasi Hermitian Hamiltonian $H$ given by (\ref{eqn:eta_quasi_Hermitian_H}), set the corresponding RHS, $\Phi\subset \mathcal{H}\subset \Phi^\prime,\Phi^\times$, and assume that $\eta$ is continuous on $(\Phi,\tau_{\Phi})$ and satisfies $\eta \Phi=\Phi$.
Then, it is shown that $\eta\otimes \eta$ is continuous on $(\Phi\hat{\otimes}\Phi,\tau_{\Phi\hat{\otimes}\Phi})$ and $\eta\otimes \eta(\Phi\hat{\otimes}\Phi)=\Phi\hat{\otimes}\Phi$\cite{Ohmori2024}.
Now we consider the RLS, ${\Phi}_\mathcal{L} \subset \mathcal{L}(\mathcal{H}) \subset 
    {\Phi}_\mathcal{L}^{\prime},{\Phi}_\mathcal{L}^{\times}$, given by (\ref{eqn:RLS}). 
Since $I_C$ is isomorphic from $(\Phi\hat{\otimes}\Phi,\tau_{\Phi\hat{\otimes}\Phi})$ onto $(\Phi_{\mathcal{L}},\tau_{\Phi_{\mathcal{L}}})$ and $\zeta$ is formed by (\ref{eqn:zeta}),  
$\zeta$ is continuous on $(\Phi_{\mathcal{L}},\tau_{\Phi_{\mathcal{L}}})$ and $\zeta \Phi_{\mathcal{L}}=\Phi_{\mathcal{L}}$ holds.
The properties indicate that we can apply the $\eta$-RHS formulation\cite{Ohmori2022} to the present RLS with respect to the positive invertible operator $\zeta$ on $\mathcal{L}(\mathcal{H})$.
Then, the following new RHS is obtained. 
\begin{eqnarray}
    {\Phi}_\mathcal{L} \subset (\mathcal{L}(\mathcal{H}))_\zeta \subset 
    {\Phi}_\mathcal{L}^{\prime},{\Phi}_\mathcal{L}^{\times},
    \label{eqn:zeta_RLS}
\end{eqnarray}
where $(\mathcal{L}(\mathcal{H}))_\zeta$ is the completion of the pre-Hilbert space $(\mathcal{L}(\mathcal{H}),\langle \cdot, \cdot \rangle_\zeta)$ composed of the set $\mathcal{L}(\mathcal{H})$ of the Liouville space and the inner product $\langle \cdot, \cdot \rangle_\zeta$ on $\mathcal{L}(\mathcal{H})$ given as 
 \begin{eqnarray}
    \langle A, B \rangle_\zeta
    = \langle A, \zeta B \rangle_\mathcal{L}~~(A,B\in \mathcal{L}(\mathcal{H})),
     \label{eqn:zeta_inner_product}
 \end{eqnarray}
and $\Phi_{\mathcal{L}}$, $\Phi_{\mathcal{L}}^\prime$, and $\Phi_{\mathcal{L}}^\times$ coincide with of (\ref{eqn:RLS}).
The RHS (\ref{eqn:zeta_RLS}) is called the $\zeta$-RLS.
Note that when introducing the $\eta\otimes \eta$-RHS derived from the RHS (\ref{eqn:tensor_pro_of_RHS}) with respect to the positive invertible operator $\eta\otimes \eta$, 
\begin{eqnarray}
    \Phi \hat{\otimes} {\Phi} \subset (\mathcal{H} \bar {\otimes} \mathcal{H})_{\eta \otimes \eta} \subset 
    (\Phi \hat{\otimes} {\Phi})^{\prime},~(\Phi \hat{\otimes} {\Phi})^{\times},
    \label{eqn:eta_eta_tensor_pro_of_RHS}
\end{eqnarray}
the $\zeta$-RLS (\ref{eqn:zeta_RLS}) is unitary equivalent to the $\eta\otimes \eta$-RHS (\ref{eqn:eta_eta_tensor_pro_of_RHS}).  

From Theorem \ref{theorem4.1}, $L$ is self-adjoint in $(\mathcal{H}\bar{\otimes}\mathcal{H})_{\eta\otimes \eta}$ and it can be verified that $L$ is continuous on $\Phi\hat{\otimes}\Phi$ with $L(\Phi\hat{\otimes}\Phi)\subset (\Phi\hat{\otimes}\Phi)$.
Therefore, the spectral expansions for $L$ can be obtained based on the $\eta \otimes\eta$-RHS (\ref{eqn:eta_eta_tensor_pro_of_RHS}).
%
%
%
The next subsection focuses on this construction.

\subsection{Spectral expansion of the quasi Hermitian $L$}
\label{sec:4.3}

By applying the general RHS formulation for an $\eta\otimes \eta$-quasi Hermitian operator developed in the previous study\cite{Ohmori2024} to the operator $L$, we can obtain the following results.
First of all, $L$ has the set of the generalized eigenvector, $\{\ket{\lambda_1}_\eta \otimes \ket{\lambda_2^C}_\eta \}$, each $\ket{\lambda_1}_\eta \otimes \ket{\lambda_2^C}_\eta$ satisfying the eigenequations,
\begin{eqnarray}
    \bra{\lambda_1}_\eta \otimes \bra{\lambda_2^C}_\eta \hat{L}=(\lambda_1-\lambda_2)\bra{\lambda_1}_\eta \otimes \bra{\lambda_2^C}_\eta,
\label{eqn:generalized_eigeneq_for_quasi_L_bra}
    \\
    \hat{L}\ket{\lambda_1}_\eta \otimes \ket{\lambda_2^C}_\eta =(\lambda_1-\lambda_2)\ket{\lambda_1}_\eta \otimes \ket{\lambda_2^C}_\eta,
\label{eqn:generalized_eigeneq_for_quasi_L_ket}
\end{eqnarray}
where $\ket{\lambda^C}_\eta=\bra{\lambda}_\eta \hat{C}$ and $\hat{L}$ is the extension given by (\ref{eqn:extension_of_L}).
Such $\hat{L}$ satisfies (\ref{eqn:L_extension_relation_to_dual_spaces}) and the following symmetric relation on $(\Phi \hat{\otimes} {\Phi})^{\prime}\cup (\Phi \hat{\otimes} {\Phi})^{\times}$, corresponding to (\ref{eqn:eta_quasi_Hermitian_L}):
    \begin{eqnarray}
    {\hat L^{\dagger}}
    =\widehat{\eta\otimes\eta}~
    \hat{L}~\widehat{\eta{\otimes}\eta}^{-1},
    \label{eqn:symmetric_relation_for_extensions_L}
    \end{eqnarray}
where $\hat{L}^\dagger$ and $\widehat{\eta\otimes\eta}$ are the extensions of the adjoint $L^\dagger$ of $L$ and $\eta \otimes \eta$, respectively; these extensions are given by replacing $L$ with $L^\dagger$ or $\eta \otimes \eta$ in (\ref{eqn:extension_of_L}). 
Note that using the representation (\ref{eqn:def_extension_dual}), $\hat{L}^\dagger$ can be characterized as
\begin{eqnarray}
    \hat{L}^\dagger=\overline{\hat{H}^\dagger\otimes\hat{I}-\hat{I}\otimes\hat{C}\hat{H}^\dagger\hat{C}},
    \label{eqn:L_dagger_extension_relation_to_dual_spaces}
\end{eqnarray}
on the dual spaces $(\Phi\hat{\otimes}\Phi)^\prime\cup (\Phi\hat{\otimes}\Phi)^\times$.
This relation is of central importance and we pause to
explain its significance in detail.

\textit{The inconsistency at the Hilbert space level.}
For a pair of non-Hermitian operators $A_i \neq A_i^\dagger$ ($i = 1, 2$)
acting on Hilbert spaces $\mathcal{H}_i$, the adjoint of their tensor product
sum satisfies only the inclusion relation\cite{Simon1980}
\begin{equation}
  (A_1^\dagger \otimes I_2 + I_1 \otimes A_2^\dagger)
  \;\subset\;
  (A_1 \otimes I_2 + I_1 \otimes A_2)^\dagger,
\end{equation}
and equality does \emph{not} hold in general. (The mismatch arises because
the domain of $(A_1 \otimes I_2 + I_1 \otimes A_2)^\dagger$ can be strictly
larger than $D(A_1^\dagger) \otimes D(A_2^\dagger)$.) 
Consequently, when the
Hamiltonian $H$ is non-Hermitian and the operator is defined as the form (\ref{eqn:L_composite_operator}), 
$L = \overline{H \otimes I - I \otimes CHC}$, its adjoint
$L^\dagger$ cannot in general be identified with
$\overline{H^\dagger \otimes I + I \otimes CH^\dagger C}$, namely, $L^\dagger\not =\overline{H^\dagger \otimes I + I \otimes CH^\dagger C}$. 
This creates a fundamental
inconsistency; one cannot define $L$ and $L^\dagger$ in a symmetric manner
that mirrors the relation $H^\dagger = \eta H \eta^{-1}$ of the Hamiltonian.

\textit{Resolution on the dual spaces.}
The relation~(\ref{eqn:L_dagger_extension_relation_to_dual_spaces}) demonstrates that this inconsistency is
\emph{resolved} by extending the operators to the dual spaces
$(\Phi\hat{\otimes}\Phi)' \cup (\Phi\hat{\otimes}\Phi)^\times$. On these
spaces, the extended adjoint $\hat{L}^\dagger$ is expressed as a proper
tensor product operator $\overline{\hat{H}^\dagger \otimes \hat{I} - \hat{I} \otimes
\hat{C}\hat{H}^\dagger\hat{C}}$, with \emph{equality} holding throughout
(rather than mere inclusion).
Note that this type of inconsistency and its resolution via dual spaces have already been discussed in previous studies on quasi-Hermitian composite systems, where related formulations can be found\cite{Ohmori2024}.

\textit{Why this matters.}
In the standard Hilbert space treatment for Hermitian systems, the
self-adjointness of $H$ ensures $L^\dagger = L$, so the inconsistency does not
arise. It is only in the non-Hermitian regime that the distinction between
``$\subset$'' and ``$=$'' becomes physically relevant. 
The RHS approach resolves
this by providing a natural domain on which the equality is restored. 
In this sense, we consider the extension to the dual spaces is not merely a mathematical convenience
but a \emph{necessity} for a consistent formulation of non-Hermitian
Liouvillian dynamics.

The spectral expansions of the bra and the ket for $L$ by $\{\ket{\lambda_1}_\eta \otimes \ket{\lambda_2^C}_\eta \}$ are represented, respectively, as follows.
\begin{eqnarray}
     \bra{\varphi}_{\mathcal{H}\bar{\otimes}\mathcal{H}}
     & = & 
    \displaystyle\int_{Sp(L)}
     \bra{\varphi}_{\mathcal{H}\bar{\otimes} \mathcal{H}}\widehat{\eta\otimes\eta}^{-1}
         \ket{\lambda_1}_{\eta}\otimes
     \ket{\lambda_2}_{\eta}
     \bra{\lambda_1}_{\eta}\otimes
     \bra{\lambda_2}_{\eta}     
    d\nu,
    \label{spectralexpansion_bra_quasi_L_a}
    \\
     \bra{\varphi}_{\mathcal{H}\bar{\otimes}\mathcal{H}}\hat{L}
     & = & 
    \displaystyle\int_{Sp(L)}
     \lambda
     \bra{\varphi}_{\mathcal{H}\bar{\otimes} \mathcal{H}}
         \ket{\lambda_1}_{\eta}\otimes
     \ket{\lambda_2}_{\eta}
     \bra{\lambda_1}_{\eta}\otimes
     \bra{\lambda_2}_{\eta}\widehat{\eta\otimes\eta}^{-1}
    d\nu,
    \label{spectralexpansion_bra_quasi_L_b}
    \\   
     \bra{\varphi}_{\mathcal{H}\bar{\otimes}\mathcal{H}}\hat{L}^\dagger
     & = & 
    \displaystyle\int_{Sp(L)}
     \lambda
     \bra{\varphi}_{\mathcal{H}\bar{\otimes} \mathcal{H}}
     \widehat{\eta\otimes\eta}^{-1}
         \ket{\lambda_1}_{\eta}\otimes
     \ket{\lambda_2}_{\eta}
     \bra{\lambda_1}_{\eta}\otimes
     \bra{\lambda_2}_{\eta}
    d\nu,
    \label{spectralexpansion_bra_quasi_L_c}
    \end{eqnarray}
and
\begin{eqnarray}
     \ket{\varphi}_{\mathcal{H}\bar{\otimes}\mathcal{H}}
     & = &
    \displaystyle\int_{Sp(L)}
     \bra{\lambda_1}_{\eta}\otimes
     \bra{\lambda_2}_{\eta}\widehat{\eta\otimes\eta}^{-1} \ket{\varphi}_{\mathcal{H}\bar{\otimes} \mathcal{H}}    
     \ket{\lambda_1}_{\eta}\otimes
     \ket{\lambda_2}_{\eta}
    d\nu.
    \label{spectralexpansion_ket_quasi_L_a}
    \\
     \hat{L}\ket{\varphi}_{\mathcal{H}\bar{\otimes}\mathcal{H}}
     & = & 
    \displaystyle\int_{Sp(L)}
     \lambda
     \bra{\lambda_1}_{\eta}\otimes
     \bra{\lambda_2}_{\eta}\ket{\varphi}_{\mathcal{H}\bar{\otimes} \mathcal{H}}    
     \widehat{\eta\otimes\eta}^{-1}
     \ket{\lambda_1}_{\eta}\otimes
     \ket{\lambda_2}_{\eta}
    d\nu,
    \label{spectralexpansion_ket_quasi_L_b}
    \\
     \hat{L}^\dagger\ket{ \varphi}_{\mathcal{H}\bar{\otimes}\mathcal{H}}
     & = & 
    \displaystyle\int_{Sp(L)}
     \lambda
     \bra{\lambda_1}_{\eta}\otimes
     \bra{\lambda_2}_{\eta}
     \widehat{\eta\otimes\eta}^{-1}
     \ket{\varphi}_{\mathcal{H}\bar{\otimes} \mathcal{H}}    
     \ket{\lambda_1}_{\eta}\otimes
     \ket{\lambda_2}_{\eta}
    d\nu.
    \label{spectralexpansion_ket_quasi_L_c}
        \end{eqnarray}
Furthermore, when introducing 
\begin{eqnarray}
    \bra{\lambda_1}_{\mathcal{H}}\otimes\bra{\lambda_2^C}_{\mathcal{H}}=\bra{\lambda_1}_{\eta}\otimes\bra{\lambda^C_2}_{\eta}\widehat{\eta\otimes\eta}^{-1}
\text{~and~} 
    \ket{\lambda_1}_{\mathcal{H}}\otimes\ket{\lambda_2^C}_{\mathcal{H}}=\widehat{\eta\otimes\eta}^{-1}\ket{\lambda_1}_{\eta}\otimes\ket{\lambda_2^C}_{\eta},
    \label{eqn:generalized_eigenvector_for_L}
\end{eqnarray}
the set $\big{\{}\ket{\lambda_1}_{\eta}\otimes\ket{\lambda_2^C}_{\eta}, \ket{\lambda_1}_{\mathcal{H}}\otimes\ket{\lambda_2^C}_{\mathcal{H}}\big{\}}_{\lambda_1-\lambda_2\in Sp(L)}$ constitutes
the complete bi-orthogonal base for the quasi-Hermitian composite system.

\subsection{Spectral expansion of the quasi Hermitian Liouville operator $\mathcal{L}_H$}
\label{sec:4.4}

Let $\hat{\zeta}$ and $\hat{\mathcal{L}}_H$ be the extensions to $\Phi_\mathcal{L}^\prime\cup \Phi_\mathcal{L}^\times$ given by replacing $L$ with  $\zeta$ and with $\mathcal{L}_H$ in (\ref{eqn:extension_of_L}), respectively, where $\zeta$ is the positive invertible operator given by (\ref{eqn:zeta}) and $\mathcal{L}_H$ is the $\zeta$-quasi Hermitian Liouville operator given as (\ref{eqn:quasi_Liouvillian}). 
Then, by applying the $\eta$-RHS framework\cite{Ohmori2022} to the $\zeta$-RLS (\ref{eqn:zeta_RLS}), the following symmetrical relation are obtained from the relation (\ref{eqn:zeta_quasi_Liouvillian}):
\begin{eqnarray}
    \hat{\mathcal{L}}_H^\dagger = \hat{\zeta} \hat{\mathcal{L}}_H\hat{\zeta}^{-1},
    \label{eqn:extension_zeta_quasi_Liouvillian}
\end{eqnarray}
where $\hat{\mathcal{L}}_H^\dagger$ is the extension of $\mathcal{L}_H^\dagger$ to $\Phi_\mathcal{L}^\prime\cup \Phi_\mathcal{L}^\times$ given by the same manner as of $\mathcal{L}_H$.
By the definition (\ref{eqn:quasi_Liouvillian}) of $\mathcal{L}_H$, the following relation between $\hat{\mathcal{L}}_H$ and $\hat{L}$ can be obtained: 
\begin{eqnarray}
\hat{\mathcal{L}}_H =\hat{I}_C^{-1}\hat{L}\hat{I}_C.
\label{eqn:relation_of_quasiHermitian_Liouville_operator_to_L}
\end{eqnarray}
Note that  (\ref{eqn:relation_of_quasiHermitian_Liouville_operator_to_L}) coincides with the relation (\ref{eqn:relation_L_and_Liouvillian}) for the Hermitian case. 
Also, since $I_C$ is the unitary operator,  $\hat{\mathcal{L}}_H^\dagger$ can be associated with $\hat{L}^\dagger$ as 
\begin{eqnarray}
\hat{\mathcal{L}}_H^\dagger =\hat{I}_C^{-1}\hat{L}^\dagger\hat{I}_C.   \label{eqn:relation_of_dagger_quasiHermitian_Liouville_operator_to_L}
\end{eqnarray}
In addition, from (\ref{eqn:zeta}), we have the relation between $\hat{\zeta}$ and $\widehat{\eta\otimes\eta}$, 
\begin{eqnarray}
\hat{\zeta}=\hat{I}_C\widehat{\eta\otimes\eta} \hat{I}_C^{-1}.
    \label{eqn:zeta_extension_relation_to_tensor_of_eta}
\end{eqnarray}

Using the generalized eigenvector $\ket{\lambda_1}_\eta \otimes \ket{\lambda_2^C}_\eta$ for $L$ given in the previous subsection, we now introduce
\begin{eqnarray}
\langle \! \langle{\lambda_1,\lambda_2}|_{\zeta} \equiv  (\bra{\lambda_1}_{\eta}\otimes\bra{\lambda_2^C}_{\eta})\hat{I}_C \text{~~and~~}
| \lambda_1,\lambda_2 \rangle \! \rangle _\zeta  \equiv  \hat{I}_C(\ket{\lambda_1}_{\eta}\otimes \ket{\lambda_2^C}_{\eta}).
\label{eqn:generalized_eigenvector_for_quasi_Liouvillian}
\end{eqnarray}
Then, they become the generalized eigenvectors for  $\mathcal{L}_H$ that satisfy the following eigenequations:
\begin{eqnarray}
    \langle \! \langle{\lambda_1,\lambda_2}|_\zeta\hat{\mathcal{L}}_H =(\lambda_1-\lambda_2)\langle \! \langle{\lambda_1,\lambda_2}|_\zeta,
    \label{eqn:generalized_eigenequations_for_quasi_Liouvillian_a}
    \\
    \hat{\mathcal{L}}_H| \lambda_1,\lambda_2 \rangle \! \rangle _\zeta=(\lambda_1-\lambda_2)| \lambda_1,\lambda_2 \rangle \! \rangle _\zeta.
    \label{eqn:generalized_eigenequations_for_quasi_Liouvillian_b}
\end{eqnarray}
Furthermore, we introduce
\begin{eqnarray}
    \langle \! \langle{\lambda_1,\lambda_2}|_{\mathcal{L}}=\langle \! \langle{\lambda_1,\lambda_2}|_{\zeta}\hat{\zeta}^{-1}
\text{~and~} 
    | \lambda_1,\lambda_2 \rangle \! \rangle _\mathcal{L}=\hat{\zeta}^{-1}| \lambda_1,\lambda_2 \rangle \! \rangle _\zeta.
    \label{eqn:generalized_eigenvector_for_Liouvillian}
\end{eqnarray}
By using  (\ref{eqn:generalized_eigenvector_for_L}), (\ref{eqn:zeta_extension_relation_to_tensor_of_eta}), and (\ref{eqn:generalized_eigenvector_for_quasi_Liouvillian}),  $\langle \! \langle{\lambda_1,\lambda_2}|_{\mathcal{L}}$ and $|{\lambda_1,\lambda_2}\rangle \! \rangle_{\mathcal{L}}$ can be associated with $\bra{\lambda_1}_{\mathcal{H}}\otimes\bra{\lambda_2^C}_{\mathcal{H}}$ and $\ket{\lambda_1}_{\mathcal{H}}\otimes\ket{\lambda_2^C}_{\mathcal{H}}$ as
\begin{eqnarray}
\langle \! \langle{\lambda_1,\lambda_2}|_{\mathcal{L}} =  (\bra{\lambda_1}_{\mathcal{H}}\otimes\bra{\lambda_2^C}_{\mathcal{H}})\hat{I}_C 
\text{~~and~~}
| \lambda_1,\lambda_2 \rangle \! \rangle _\mathcal{L}  =  \hat{I}_C(\ket{\lambda_1}_{\mathcal{H}}\otimes \ket{\lambda_2^C}_{\mathcal{H}}),
\label{eqn:generalized_eigenvector_for_quasi_Liouvillian_new}
\end{eqnarray}
respectively.
Since $\big{\{}\ket{\lambda_1}_{\eta}\otimes\ket{\lambda_2^C}_{\eta}, \ket{\lambda_1}_{\mathcal{H}}\otimes\ket{\lambda_2^C}_{\mathcal{H}}\big{\}}_{\lambda_1-\lambda_2\in Sp(L)}$ is the complete bi-orthogonal base, from (\ref{eqn:generalized_eigenvector_for_quasi_Liouvillian}) and (\ref{eqn:generalized_eigenvector_for_quasi_Liouvillian_new}), we have
\begin{eqnarray}
 \langle A, B \rangle_\mathcal{L} & = &   \langle I_C^{-1}A, I_C^{-1}B \rangle_{\mathcal{H}\bar{\otimes} \mathcal{H}}
        \nonumber\\
        & = &  \displaystyle\int_{Sp(L)}
     \bra{I_C^{-1}A}_{\mathcal{H}\bar{\otimes} \mathcal{H}}
         \ket{\lambda_1}_{\eta}\otimes
     \ket{\lambda_2}_{\eta}
     \bra{\lambda_1}_{\mathcal{H}}\otimes
     \bra{\lambda_2}_{\mathcal{H}}\ket{I_C^{-1}B}_{\mathcal{H}\bar{\otimes}\mathcal{H}}     
    d\nu
    \nonumber\\
        & = &  \displaystyle\int_{Sp(\mathcal{L}_H)}
     \bra{A}_{\mathcal{L}}\hat{I}_C
         (\ket{\lambda_1}_{\eta}\otimes
     \ket{\lambda_2}_{\eta})(
     \bra{\lambda_1}_{\mathcal{H}}\otimes
     \bra{\lambda_2}_{\mathcal{H}})\hat{I}_C\ket{B}_{\mathcal{L}}     
    d\nu
    \nonumber \\
   & = &  \displaystyle\int_{Sp(\mathcal{L}_H)}
     \bra{A}_{\mathcal{L}}| \lambda_1,\lambda_2 \rangle \! \rangle _\zeta
       \langle \! \langle \lambda_1,\lambda_2 |_\mathcal{L}\ket{B}_{\mathcal{L}}     
    d\nu,
    \label{eqn:bi_o.n.b.calculation}
\end{eqnarray}
for any $A,B\in \Phi_{\mathcal{L}}$.
Therefore, we obtain the completion relation, 
\begin{eqnarray}
    I =  \displaystyle\int_{Sp(\mathcal{L}_H)}
     | \lambda_1,\lambda_2 \rangle \! \rangle _\zeta
       \langle \! \langle \lambda_1,\lambda_2 |_\mathcal{L}d\nu
    =
    \displaystyle\int_{Sp(\mathcal{L}_H)}| \lambda_1,\lambda_2 \rangle \! \rangle _\mathcal{L}
       \langle \! \langle \lambda_1,\lambda_2 |_\zeta
    d\nu.
    \label{eqn:completion}
\end{eqnarray}
Similarly, we can show the bi-orthogonal relation  
\begin{eqnarray}
\langle \! \langle \lambda_1^\prime,\lambda_2^\prime |_\zeta|\lambda_1,\lambda_2 \rangle \! \rangle _\mathcal{L}   
=
\langle \! \langle \lambda_1^\prime,\lambda_2^\prime |_\mathcal{L}|\lambda_1,\lambda_2 \rangle \! \rangle _\zeta   
    =
    \Check{\delta}(\lambda_1^\prime-\lambda_1)
    \Check{\delta}(\lambda_2^\prime-\lambda_2).
    \label{eqn:bi-o.g.r}
\end{eqnarray}
Thus, from (\ref{eqn:completion}) and (\ref{eqn:bi-o.g.r}), the set $\big{\{}|\lambda_1,\lambda_2 \rangle \! \rangle _\zeta,|\lambda_1,\lambda_2 \rangle \! \rangle _\mathcal{L}\big{\}}_{\lambda_1-\lambda_2\in Sp(\mathcal{L})}$ consists of the complete bi-orthogonal base.
Corresponding to the spectral expansions (\ref{spectralexpansion_bra_quasi_L_a})-(\ref{spectralexpansion_ket_quasi_L_c}) for $L$, the spectral expansions for the $\zeta$-quasi Hermitian Liouville operator $\mathcal{L}_H$ can be described as follows: for $A\in \Phi_\mathcal{L}$,
\begin{eqnarray}
     \bra{A}_{\mathcal{L}}
     & = & 
    \displaystyle\int_{Sp(\mathcal{L}_H)}
     \bra{A}_{\mathcal{L}}\hat{\zeta}^{-1}|
     \lambda_1,\lambda_2 \rangle \! \rangle _\zeta \langle \! \langle{\lambda_1,\lambda_2}| _\zeta
     d\nu,
    \label{spectralexpansion_bra_quasi_Liouvillian_a}
    \\
     \bra{A}_{\mathcal{L}}\hat{\mathcal{L}}_H
     & = & 
    \displaystyle\int_{Sp(\mathcal{L}_H)}
     \lambda \bra{A}_{\mathcal{L}}|
     \lambda_1,\lambda_2 \rangle \! \rangle _\zeta \langle \! \langle{\lambda_1,\lambda_2}| _\zeta \hat{\zeta}^{-1}
     d\nu,
    \label{spectralexpansion_bra_quasi_Liouvillian_b}
    \\   
     \bra{A}_{\mathcal{L}}\hat{\mathcal{L}}_H^\dagger
     & = & 
    \displaystyle\int_{Sp(\mathcal{L}_H)}
     \lambda \bra{A}_{\mathcal{L}}\hat{\zeta}^{-1}|
     \lambda_1,\lambda_2 \rangle \! \rangle _\zeta \langle \! \langle{\lambda_1,\lambda_2}| _\zeta 
     d\nu,
    \label{spectralexpansion_bra_quasi_L_c}
    \end{eqnarray}
and
\begin{eqnarray}
     \ket{A}_{\mathcal{L}}
     & = &
    \displaystyle\int_{Sp(\mathcal{L}_H)}
     \langle \! \langle{\lambda_1,\lambda_2}| _\zeta\hat{\zeta}^{-1} \ket{A}_{\mathcal{L}}|
     \lambda_1,\lambda_2 \rangle \! \rangle _\zeta 
    d\nu,
    \label{spectralexpansion_ket_quasi_Liouvillian_a}
    \\
     \hat{\mathcal{L}}_H\ket{A}_{\mathcal{L}}
     & = &
    \displaystyle\int_{Sp(\mathcal{L}_H)}
     \lambda \langle \! \langle{\lambda_1,\lambda_2}| _\zeta \ket{A}_{\mathcal{L}}\hat{\zeta}^{-1}|
     \lambda_1,\lambda_2 \rangle \! \rangle _\zeta 
    d\nu,
    \label{spectralexpansion_ket_quasi_Liouvillian_b}
    \\
    \hat{\mathcal{L}}_H^\dagger \ket{A}_{\mathcal{L}}
     & = &
    \displaystyle\int_{Sp(\mathcal{L}_H)}
     \lambda \langle \! \langle{\lambda_1,\lambda_2}| _\zeta \hat{\zeta}^{-1}\ket{A}_{\mathcal{L}}|
     \lambda_1,\lambda_2 \rangle \! \rangle _\zeta 
    d\nu.
    \label{spectralexpansion_ket_quasi_Liouvillian_c}
        \end{eqnarray}

\subsection{Application to non-Hermitian Harmonic Oscillator}
\label{sec:4.5}

The final topic of the present study is the application of our formulation to non-Hermitian Liouvillian system.
We now focus on the following simple non-Hermitian 
$\mathcal{PT}$-symmetric oscillator Hamiltonian proposed by Swanson\cite{Swanson2004},
\begin{eqnarray}
    H=\omega \left(a^{\dagger}a+\frac{1}{2}\right ) +\alpha a^2+\beta a^{\dagger,2},
    \label{eqn:Swanson}
\end{eqnarray}
where $\omega,\,\alpha,\,\beta \in \mathbb{R}$, with $\alpha\neq\beta$.
$a$ and $a^\dagger$ are ladder operators that are represented as
\begin{equation}
     \label{eqn:a_to_xp}
     a =
    \frac{1}{\sqrt{2}}\left( \frac{x}{\ell} + i\frac{\ell}{\hbar} p \right), \quad
    \hat a^\dagger = \frac{1}{\sqrt{2}}\left( \frac{x}{\ell} - i\frac{\ell}{\hbar} p \right),
\end{equation}
where $\ell=\sqrt{\hbar/m\omega}$ is the typical length for the ordinary harmonic oscillator with the frequency $\omega$ and mass $m$.
Here we assume $\omega^2-4\alpha\beta>0$.
Then, the following similarity transformation can be performed to derive an ordinary hermitian oscillator Hamiltonian\cite{Quesne2007}: 
\begin{eqnarray}
h_\rho & = &  \rho H\rho^{-1}\nonumber\\
        & = & \frac{1}{2M(z)}  p^2 + \frac{1}{2} M(z)\Omega^2 {x}^2,   
    \label{eqn:5_similarity_tra}
\end{eqnarray}
where
$z:=2\kappa/\chi$ with $z\in[-1, 1]$
is the free parameter that determines the Hermiticity of $h_\rho$,
$\Omega^2 \, = \,\omega^2 - 4\alpha\beta>0$, 
\begin{equation}
    M^{-1}(z) = \frac{-z(\alpha+\beta)+\omega-(\alpha+\beta-z\omega)\sqrt{1-\frac{(1-z^2)(\alpha-\beta)^2}{(\alpha+\beta-z\omega)^2}}}{(1+z)\hbar \ell^{-2}},    
\end{equation}
and 
\begin{eqnarray}
        \eta(x) = \rho^2(x)
=\exp \left(
        -\frac{\alpha-\beta}{\omega-\alpha-\beta} \frac{{x}^2}{\ell^2}
        \right).
\label{eqn:eta}
\end{eqnarray}
Note that since $H$ connects to the Hermitian Hamiltonian $h_\rho$ by the similarity transformation with respect to $\rho$, $H$ becomes an $\eta$-quasi Hermitian Hamiltonian that satisfies (\ref{eqn:eta_quasi_Hermitian_H}), $H^\dagger=\eta H\eta^{-1}$.

We now set the RHS; since $h_\rho$ has the same form as (\ref{eqn:Harmonic_osc_Hamiltonian}), the RHS for $h_{\rho}$ can be provided by (\ref{eqn:RHS_for_Harmonic_oscillator}), $\Phi_{ho}\subset\mathcal{H}_{ho}\subset\Phi_{ho}^\prime,\Phi_{ho}^\times$ and its generalized eigenvectors are described as $\ket{n}\in \Phi_{ho}^\times$~($n=0,1,2,\cdots$).
Furthermore, for $H$, the generalized eigenvectors can be obtained as $\ket{n}_{\eta}\equiv\hat{\rho}^{-1}\ket{n}$.
Considering the conjugation $C_{ho}$ given in Sec. \ref{sec:3.4}, by Theorem \ref{theorem4.1}, the operator $L$ of the form (\ref{eqn:L_composite_operator}), $L=\overline{H\otimes I -I\otimes C_{ho}HC_{ho}}=\overline{H\otimes I -I\otimes H}$, becomes the $\eta\otimes \eta$-quasi Hermitian operator.
The spectral expansions for $L$ is then expressed using (\ref{spectralexpansion_bra_quasi_L_a})-(\ref{spectralexpansion_ket_quasi_L_c}), as follows: for $\varphi\in \Phi_{ho}\hat{\otimes}\Phi_{ho}$,
\begin{eqnarray}
     \bra{\varphi}_{\mathcal{H}_{ho}\bar{\otimes}\mathcal{H}_{ho}}
     & = & 
    \displaystyle\sum_{m,n}
     \bra{\varphi}_{\mathcal{H}_{ho}\bar{\otimes} \mathcal{H}_{ho}}\widehat{\eta\otimes\eta}^{-1}
         \ket{m}_{\eta}\otimes
     \ket{n}_{\eta}
     \bra{m}_{\eta}\otimes
     \bra{n}_{\eta},
    \label{eqn:spectralexpansion_bra_Swanson_L_a}
    \\
     \bra{\varphi}_{\mathcal{H}_{ho}\bar{\otimes}\mathcal{H}_{ho}}\hat{L}
     & = & 
    \displaystyle\sum_{m,n}\hbar \omega (m-n)
     \bra{\varphi}_{\mathcal{H}_{ho}\bar{\otimes} \mathcal{H}_{ho}}\widehat{\eta\otimes\eta}^{-1}
         \ket{m}_{\eta}\otimes
     \ket{n}_{\eta}
     \bra{m}_{\eta}\otimes
     \bra{n}_{\eta},
    \label{eqn:spectralexpansion_bra_Swanson_L_b}
    \\   
 \bra{\varphi}_{\mathcal{H}_{ho}\bar{\otimes}\mathcal{H}_{ho}}\hat{L}^\dagger
     & = & 
    \displaystyle\sum_{m,n}\hbar \omega (m-n)
     \bra{\varphi}_{\mathcal{H}_{ho}\bar{\otimes} \mathcal{H}_{ho}}\widehat{\eta\otimes\eta}^{-1}
         \ket{m}_{\eta}\otimes
     \ket{n}_{\eta}
     \bra{m}_{\eta}\otimes
     \bra{n}_{\eta},
    \label{eqn:spectralexpansion_bra_Swanson_L_c}
    \end{eqnarray}
and
\begin{eqnarray}
     \ket{\varphi}_{\mathcal{H}_{ho}\bar{\otimes}\mathcal{H}_{ho}}
     & = &
    \displaystyle\sum_{m,n}
     \bra{m}_{\eta}\otimes
     \bra{n}_{\eta}\widehat{\eta\otimes\eta}^{-1} \ket{\varphi}_{\mathcal{H}_{ho}\bar{\otimes} \mathcal{H}_{ho}}    
     \ket{m}_{\eta}\otimes
     \ket{n}_{\eta},
    \label{spectralexpansion_ket_Swanson_L_a}
    \\
          \hat{L}\ket{\varphi}_{\mathcal{H}_{ho}\bar{\otimes}\mathcal{H}_{ho}}
     & = &
    \displaystyle\sum_{m,n}\hbar \omega (m-n)
     \bra{m}_{\eta}\otimes
     \bra{n}_{\eta}\widehat{\eta\otimes\eta}^{-1} \ket{\varphi}_{\mathcal{H}_{ho}\bar{\otimes} \mathcal{H}_{ho}}    
     \ket{m}_{\eta}\otimes
     \ket{n}_{\eta},
    \label{spectralexpansion_ket_Swanson_L_b}
    \\
          \hat{L}^\dagger\ket{\varphi}_{\mathcal{H}_{ho}\bar{\otimes}\mathcal{H}_{ho}}
     & = &
    \displaystyle\sum_{m,n}\hbar \omega (m-n)
     \bra{m}_{\eta}\otimes
     \bra{n}_{\eta}\widehat{\eta\otimes\eta}^{-1} \ket{\varphi}_{\mathcal{H}_{ho}\bar{\otimes} \mathcal{H}_{ho}}    
     \ket{m}_{\eta}\otimes
     \ket{n}_{\eta}.
    \label{spectralexpansion_ket_Swanson_L_c}
        \end{eqnarray}
Here, $\hat{L}$ and $\hat{L}^\dagger$ are the extension of $L$ and $L^\dagger$ to $(\Phi_{ho}\hat{\otimes}\Phi_{ho})^\prime\cup (\Phi_{ho}\hat{\otimes}\Phi_{ho})^\times$ that are expressed as the following symmetric forms by (\ref{eqn:L_extension_relation_to_dual_spaces}) and (\ref{eqn:L_dagger_extension_relation_to_dual_spaces}): 
 \begin{eqnarray} 
\hat{L}=\overline{\hat{H}\otimes I-I\otimes \hat{H}} \text{~~and~~} \hat{L}^\dagger=\overline{\hat{H}^\dagger\otimes I-I\otimes \hat{H}^\dagger},
\label{eqn:Swanson_L_daggerL}
 \end{eqnarray}
where $\hat{H}$ and $\hat{H}^\dagger$ are the extension of $H$ given as (\ref{eqn:Swanson}) and its adjoint $H^\dagger$ to $\Phi_{ho}^\prime\cup \Phi_{ho}^\times$, respectively.
Note that $\hat{L}$ and $\hat{L}^\dagger$ are connected symmetrically by (\ref{eqn:symmetric_relation_for_extensions_L}), ${\hat L^{\dagger}}
    =\widehat{\eta\otimes\eta}~
    \hat{L}~\widehat{\eta{\otimes}\eta}^{-1}$.

When introducing $\zeta$ by (\ref{eqn:zeta}), from Theorem \ref{theorem4.2}, the Liouville operator $\mathcal{L}_H$ corresponding to the present non-Hermitian Hamiltonian $H$ given as (\ref{eqn:Swanson}) becomes the $\zeta$-quasi Hermitian operator.
Then, the extensions $\hat{\zeta}$, $\hat{\mathcal{L}}_H$, and $\hat{\mathcal{L}}_H^\dagger$ satisfy the relations (\ref{eqn:extension_zeta_quasi_Liouvillian})-(\ref{eqn:zeta_extension_relation_to_tensor_of_eta}).
By  (\ref{eqn:generalized_eigenvector_for_quasi_Liouvillian}), 
the generalized eigenvectors for $\mathcal{L}_H$ can be constructed using the eigenvector $\ket{m}_\eta\otimes\ket{n}_\eta$ for $L$ and $\hat{I}_{C_{ho}}$, as follows:
\begin{eqnarray}
\langle \! \langle{m,n}|_{\zeta} \equiv  (\bra{m}_{\eta}\otimes\bra{n}_{\eta})\hat{I}_{C_{ho}} \text{~~and~~}
| m,n \rangle \! \rangle _\zeta  \equiv  \hat{I}_{C_{ho}}(\ket{m}_{\eta}\otimes \ket{n}_{\eta}).
\label{eqn:generalized_eigenvector_for_Swanson_Liouvillian}
\end{eqnarray}
Then, they provide the eigenequations,
\begin{eqnarray}
    \langle \! \langle{m,n}|_\zeta\hat{\mathcal{L}}_H =\hbar\omega(m-n)\langle \! \langle{m,n}|_\zeta,
    \label{eqn:generalized_eigenequations_for_Swanson_Liouvillian_a}
    \\
    \hat{\mathcal{L}}_H| m,n \rangle \! \rangle _\zeta=\hbar\omega(m-n)| m,n \rangle \! \rangle _\zeta.
    \label{eqn:generalized_eigenequations_for_Swanson_Liouvillian_b}
\end{eqnarray}
Therefore, the spectral expansions of the super bra and ket for the $\zeta$-quasi Hermitian Liouville operator $\mathcal{L}_H$ with respect to $H$ are represented as follows:
for $A\in \Phi_\mathcal{L}$,
\begin{eqnarray}
     \bra{A}_{\mathcal{L}}
     & = & 
    \displaystyle\sum_{m,n}
     \bra{A}_{\mathcal{L}}\hat{\zeta}^{-1}|
     m,n \rangle \! \rangle _\zeta \langle \! \langle{m,n}| _\zeta,
    \label{spectralexpansion_bra_Swanson_Liouvillian_a}
    \\
     \bra{A}_{\mathcal{L}}\hat{\mathcal{L}}_H
     & = & 
    \displaystyle\sum_{m,n}
     \hbar\omega(m-n) \bra{A}_{\mathcal{L}}|
     m,n \rangle \! \rangle _\zeta \langle \! \langle{m,n}| _\zeta \hat{\zeta}^{-1},
    \label{spectralexpansion_bra_Swanson_Liouvillian_b}
    \\   
     \bra{A}_{\mathcal{L}}\hat{\mathcal{L}}_H^\dagger
     & = & 
    \displaystyle\sum_{m,n}
     \hbar\omega(m-n) \bra{A}_{\mathcal{L}}\hat{\zeta}^{-1}|
     m,n \rangle \! \rangle _\zeta \langle \! \langle{m,n}| _\zeta,
    \label{spectralexpansion_bra_Swanson_L_c}
    \end{eqnarray}
and
\begin{eqnarray}
     \ket{A}_{\mathcal{L}}
     & = &
    \displaystyle\sum_{m,n}
     \langle \! \langle{m,n}| _\zeta\hat{\zeta}^{-1} \ket{A}_{\mathcal{L}}|
     m,n \rangle \! \rangle _\zeta,
    \label{spectralexpansion_ket_Swanson_Liouvillian_a}
    \\
     \hat{\mathcal{L}}_H\ket{A}_{\mathcal{L}}
     & = &
    \displaystyle\sum_{m,n}
     \hbar\omega(m-n) \! \langle{m,n}| _\zeta \ket{A}_{\mathcal{L}}\hat{\zeta}^{-1}|
     m,n \rangle \! \rangle _\zeta,
    \label{spectralexpansion_ket_Swanson_Liouvillian_b}
    \\
    \hat{\mathcal{L}}_H^\dagger \ket{A}_{\mathcal{L}}
     & = &
    \displaystyle\sum_{m,n}
     \hbar\omega(m-n)\langle \! \langle{m,n}| _\zeta \hat{\zeta}^{-1}\ket{A}_{\mathcal{L}}|
     m,n \rangle \! \rangle _\zeta.
    \label{spectralexpansion_ket_Swanson_Liouvillian_c}
        \end{eqnarray}
Note that it is easy to transform the obtained results to other representations.

In comparing the obtained spectral
expansions~(\ref{spectralexpansion_bra_Swanson_Liouvillian_a})-(\ref{spectralexpansion_ket_Swanson_Liouvillian_c}) with
those of the Hermitian
case~(\ref{eqn:HO_expansions_of_bra_by_Liouvillian_a})-(\ref{eqn:HO_expansions_of_ket_by_Liouvillian_b}), the following two essential
differences are found.
First, in the quasi-Hermitian case, the spectral expansions require the
insertion of the inverse metric operator $\hat{\zeta}^{-1}$, reflecting the
positive-definite metric of the Liouville space induced from the inner product
$\langle A, B \rangle_\zeta = \langle A, \zeta B \rangle_\mathcal{L}$. 
In the
Hermitian limit ($\eta \to I$, hence $\zeta \to I$), the metric operator
reduces to the identity, and hence the expansions recover the Hermitian forms.
Second, the set of generalized eigenvectors for the quasi-Hermitian Liouville operator
forms a complete bi-orthogonal system, (\ref{eqn:completion}) and (\ref{eqn:bi-o.g.r}), rather than the single orthonormal base of the Hermitian case.
This bi-orthogonality is a direct consequence of the
fact that $\mathcal{L}_H$ and $\mathcal{L}_H^\dagger$ share the same eigenvalues but possess
distinct generalized eigenvectors, $|
     m,n \rangle \! \rangle _\zeta$ and $\zeta|
     m,n \rangle \! \rangle _\zeta$, respectively. 
Note that the eigenvalues $\hbar\omega(m - n)$ remain real and coincide with those
of the Hermitian case, since the Swanson Hamiltonian is quasi-Hermitian with a
real spectrum.

Although the RLS formulation discussed in this study focuses on quasi-Hermitian Liouville operators, this framework can be widely used for general superoperators on Liouville space.
Our present study contains various applications
since the description of quantum systems by the non-Hermitian Liouville operator and Liouville space has been extensively developed especially in the modern quantum theory, such as open quantum systems\cite{Gyamfi2020,Polonyi2021,Sukharnikov2023}.
And extending the present RLS formalism to other classes of non-Hermitian operators, such as pseudo-Hermitian operators with complex eigenvalues, constitutes an attractive and important direction for future investigation.

\section{Conclusion}
\label{sec:6}

A mathematical description of Dirac's bra-ket formalism for Liouville space, namely, the super bra-ket formalism, with respect to quasi-Hermitian systems is investigated based on RHS approach. 
The RLS that is the RHS suitable for Liouville space is reconstructed by making use of the mathematical unitarily transformation between the Liouville space and the tensor product of RHS.
The obtained RLS provides a rigorous mathematical framework for description of the super bra-ket vectors as well as the spectral expansions for both Hermitian and quasi-Hermitian Liouvillian operators. 
They are formulated and executed in the dual and anti-dual spaces.
Based on the formulation under these dual spaces, we show that the non-Hermitian Liouville operator and its adjoint can be symmetrically constructed while retaining their symmetric structure.
As the example of our present formulation, the Hermitian and non-Hermitian harmonic oscillators are focused on to clarify the essential differences of the spectral expansions between Hermitian and non-Hermitian models. 
%




\bigskip

\noindent
{\bf Acknowledgement}

The authors are grateful to
Prof. Y.~Yamazaki, Prof. T.~Yamamoto, Prof. Emeritus Y.~Yamanaka, and Prof. Emeritus A.~Kitada for worthwhile comments and encouragement.
We also appreciate 
Prof. H.~Ujino, Prof. I.~Sasaki, 
Prof. H.~Saigo, Prof. F.~Hiroshima, 
Prof. S. Matsutani for comments and encouragement.
%
This work was supported by the JSPS KAKENHI Grant Numbers 26K06948.
%

\noindent
{\bf Data Availability}

Data sharing is not applicable to this article as no new data were created or analyzed in this study.

\bigskip


\end{document}